\documentclass[11pt]{article}

\usepackage[utf8]{inputenc} 
\usepackage[T1]{fontenc}    
\usepackage{hyperref}       
\usepackage{url}            
\usepackage{booktabs}       
\usepackage{amsfonts}       
\usepackage{nicefrac}       
\usepackage{microtype}      
\usepackage{color}

\usepackage{graphicx} 

\usepackage{amsthm}
\usepackage{amssymb}
\usepackage{algorithm}
\usepackage[noend]{algpseudocode}
\usepackage{nicefrac}

\usepackage[margin=1in]{geometry}

\usepackage{amsmath}
\usepackage{thmtools, thm-restate}
\usepackage{multirow}
\usepackage{tabularx}
\usepackage{colortbl}
\usepackage{color}
\usepackage[small]{caption}

\newtheorem{theorem}{Theorem}[section]
\newtheorem{lemma}[theorem]{Lemma}
\newtheorem{fact}[theorem]{Fact}

\newtheorem{corollary}[theorem]{Corollary}

\newtheorem{definition}[theorem]{Definition}

\newtheorem{assumption}{Assumption}

\newcommand{\FT}{{\mathfrak{F}}}
\newcommand{\N}{{\mathbb{N}}}

\newcommand{\R}{{\mathbb{R}}}

\newcommand{\beq}{\begin{eqnarray}}
\newcommand{\eeq}{\end{eqnarray}}

\newcommand{\sign}{\text{sign}}

\DeclareMathOperator{\Tr}{Tr}

\DeclareMathOperator*{\expt}{\mathbb{E}}

\newcommand{\EE}[2]{{\expt_{#1}{#2}}}

\newcommand{\pd}[2]{\frac{\partial#1}{\partial#2}}
\newif\ifshort
\shorttrue

\def\showauthornotes{0}

\def\showsupmaterial{0}

\def\abs#1{\left| #1 \right|}
\newcommand{\norm}[1]{\ensuremath{\left\lVert #1 \right\rVert}}

\newcommand\rea{\mathbb R}

\newcommand{\marginlabel}[1]%
{\mbox{}\marginpar{\it{\raggedleft\hspace{0pt}#1}}}
\newcommand{\poly}{\mathrm{poly}}

\definecolor{Mygray}{gray}{0.8}

 \ifcsname ifcommentflag\endcsname\else
  \expandafter\let\csname ifcommentflag\expandafter\endcsname
                  \csname iffalse\endcsname
\fi

\ifnum\showauthornotes=1

\else

\fi

\ifnum\showauthornotes=1
\newcommand{\Authornote}[2]{{\small\color{red}{[#1: #2]}}}
\else
\newcommand{\Authornote}[2]{}
\fi

\newcommand{\supmaterial}[1]{\ifnum\showsupmaterial =1 supplementary material \else #1 \fi}

\newcommand{\Anote}{\Authornote{A}}

\title{
Convergence Results for Neural Networks
 via Electrodynamics
}

\author{
  Rina Panigrahy \\
  Google Inc. \\
  Mountain View, CA\\
  \texttt{rinap@google.com} \and
  Sushant Sachdeva\thanks{This work was done when the author was a Research Scientist at Google, Mountain View, CA.} \\
  University of Toronto \\
  Toronto, Canada\\
  \texttt{sachdeva@cs.toronto.edu} \and Qiuyi Zhang\thanks{Part of
    this work was done when the author was
    an intern at Google, Mountain View, CA.}  \\
  University of California Berkeley, \\
  Berkeley, CA \\
  \texttt{10zhangqiuyi@berkeley.edu} }

\begin{document} 

\maketitle

\begin{abstract} 
We study whether a depth two neural network can learn another 
depth two network using gradient descent.
Assuming a linear output node,
we show that
the question of whether gradient descent converges to the 
target function is equivalent to the following question in
electrodynamics: 
Given $k$ fixed protons in $\rea^d,$ and $k$ electrons,
each moving due to the attractive force from the protons and repulsive
force from the remaining electrons,
whether at equilibrium all the electrons will be matched up with
the protons, up to a permutation. 
Under the standard electrical
force, this follows from the classic Earnshaw's theorem. In our setting,
the force  is 
determined by the activation function and the
input distribution.  
Building on this equivalence, we prove the
existence of an activation function such that 
gradient descent learns
at least one of the
hidden nodes in the target network. 
Iterating, we show that gradient
descent can be used to learn the entire network one node at a time.
\end{abstract} 

\section{Introduction}

Deep learning has resulted in major strides in machine learning applications including speech recognition, image classification, and ad-matching. The simple idea of using multiple layers of nodes with a non-linear activation function at each node allows one to express any function.  To learn a certain target function we just use (stochastic) gradient descent to minimize the loss; this approach has resulted in significantly lower error rates for several real world functions, such as those in the above applications. Naturally the question remains: how close are we to the optimal values of the network weight parameters? Are we stuck in some bad local minima? While there are several recent works \cite{ChoromanskaHMAL14, DauphinPGCGB14, Kawaguchi16a} that have tried to study the presence of local minima, the picture is far from clear.

There has been some work on studying how well can neural networks
learn some synthetic function classes
(e.g. polynomials~\cite{valiant2014learning}, decision trees).  In
this work we study how well can neural networks learn neural networks
with gradient descent?  Our focus here, via the framework of proper
learning, is to understand if a neural network can learn a function
from the same class (and hence achieve vanishing error).

%
%

Specifically, if the target function is a neural network with randomly
initialized weights, and we attempt to learn it using a network with
the same architecture, then, will gradient descent converge to the
target function?


Experimental simulations (see Figure~\ref{expconverge} and
Section~\ref{experiments} for further details) show that for depth 2
networks of different widths, with random network weights, stochastic
gradient descent of a hypothesis network with the same architecture
converges to a squared $\ell_2$ error that is a small percentage of a
random network, indicating that SGD can learn these shallow networks
with random weights. Because our activations are sigmoidal from -1 to
1, the training error starts from a value of about $1$ (random
guessing) and diminishes quickly to under $0.002$. This seems to hold
even when the width, the number of hidden nodes, is substantially
increased (even up to 125 nodes), but depth is held constant at $2$.

In this paper, we attempt to understand this phenomenon
theoretically. We prove that, under some assumptions, depth-2 neural
networks can learn functions from the same class with vanishingly
small error using gradient descent.

\begin{figure}[h]
\centering
\includegraphics[width = 4.5in]{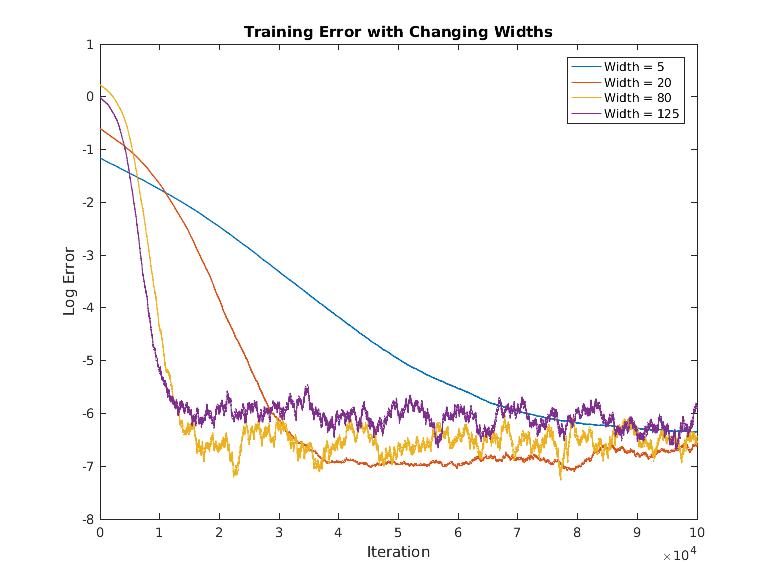}
\caption{Test Error of Depth 2 Networks of Varying Width.}
\label{expconverge}
\end{figure}


\subsection{Results and Contributions.} 
We theoretically investigate the question of convergence for networks of depth two.
Our main conceptual contribution is that for depth $2$ networks where the top node is a sum node, the question of whether gradient descent converges to the desired target function is equivalent to the following question in electrodynamics: Given $k$ fixed protons in $\rea^d,$ and $k$ moving electrons,
with all the electrons moving under the influence of the 
electrical force of attraction from the protons and repulsion from the remaining electrons,
at equilibrium, are all the electrons matched up with all the fixed protons, up to a permutation?  

In the above, $k$ is the number of hidden units, $d$ is the number of inputs, the positions of each fixed charge is the input weight vector of
a hidden unit in the target network, and the initial positions of the moving charges are the initial values of the weight vectors for the hidden units in the learning network. The motion of the charges essentially tracks the change in the network during gradient descent. The force between a pair of charges is not given by the standard electrical force of $1/r^2$ (where $r$ is the distance between the charges), but by a function determined by the activation and the input distribution. Thus the question of convergence in these simplified depth two networks can be resolved by studying the equivalent electrodynamics question with the corresponding force function.
%

\begin{theorem}[informal statement of Theorem~\ref{EPDyn}]
Applying gradient descent for learning the output of a depth two network
with $k$ hidden units with activation $\sigma,$ and a linear output
node, under squared loss, using a network of the same architecture,
is equivalent to the motion of $k$ charges in the presence of $k$ fixed charges where the force between each pair of charges is given by a potential function that depends on $\sigma$ and the input distribution.  \end{theorem}
%
Based on this correspondence we  prove the existence of an
activation function such that the corresponding gradient descent
dynamics under standard Gaussian inputs result in learning at least
one of the hidden nodes in the target network. We then show that this
allows us to learn the complete target network one node at a time. For more realistic activation functions, we only obtain
partial results. We assume the sample complexity
is close to its infinite limit.
%

%
\begin{theorem}[informal statement of Theorem~\ref{almostHarmSGD}]
There is an activation function such that running gradient
  descent for minimizing the squared loss along with $\ell_2$
  regularization for standard Gaussian inputs, at convergence, 
  we learn at least one of
  the hidden weights of the target neural network.
\end{theorem}
We  prove that the above result can be iterated to learn the entire network node-by-node using gradient descent (Theorem~\ref{nodeWise}).  Our algorithm learns a network with the same architecture and number of hidden nodes as the target network, 
in contrast with several existing improper learning results.

In the appendix, we show some weak results for more practical activations. For the sign activation, we show that for the loss with respect to a single node, the only local minima are at the hidden target nodes with high probability if the target network has a randomly picked top layer. For the polynomial activation, we derive a similar result under the assumption that the hidden nodes are orthonormal.

 \begin{table}[h!]
 \noindent
 \vskip 0.1in
 \begin{center}
 \begin{small}
 \begin{sc}
 \begin{tabular}{
   |p{\dimexpr.28\linewidth-2\tabcolsep-1.3333\arrayrulewidth}
   |p{\dimexpr.30\linewidth-2\tabcolsep-1.3333\arrayrulewidth}
   |p{\dimexpr.42\linewidth-2\tabcolsep-1.3333\arrayrulewidth}|
   }
    \hline 
         Name of Activation&  Potential  ($\Phi(\theta,w)$)    & Convergence? \\ \hline 
        Almost   $\lambda$-harmonic  & Complicated  (see Lem~\ref{almostHarmReal}) & Yes, Thm~\ref{nodeWise}\\
         Sign & $1 - \frac{2}{\pi}\cos^{-1}(\theta^Tw)$       & Yes for d = 2, Lem~\ref{SignConv} \\   
Polynomial  & $(\theta^Tw)^m$       & Yes, for orthonormal $w_i.$ Lem~\ref{PolyStrict}  \\        
         \hline
 \end{tabular}
 \end{sc}
 \end{small}
 \end{center}
 \caption{Activation, Potentials, and Convergence Results Summary}
 \label{table1}
 \end{table} 

\subsection{Intuition and Techniques.}
Note that for the standard electric potential function given by $\Phi = 1/r$ where $r$ is the distance between the charges, it is known from Earnshaw's theorem that an electrodynamic system with some fixed protons and some moving electrons is at equilibrium only when the moving electrons coincide with the fixed protons. Given our translation above between electrodynamic systems and depth 2 networks (Section~\ref{sec:epdyn}), this would imply learnability of depth 2 networks under gradient descent under $\ell_2$ loss, if the activation function corresponds to the electrostatic potential. However, there exists no activation function $\sigma$ corresponding to this $\Phi$.
%


The proof of Earnshaw's theorem is based on the fact that the electrostatic potential is harmonic, \emph{i.e}, its Laplacian (trace of its Hessian) is identically zero. This ensures that at every critical point, there is direction of potential reduction (unless the hessian is identically zero). We generalize these ideas to potential functions that are eigenfunctions of the Laplacians, $\lambda$-harmonic potentials (Section~\ref{sec:earnshaw}). However, these potentials are unbounded. Subsequently, we construct a non-explicit activation function such that the corresponding potential is bounded and is almost $\lambda$-harmonic, \emph{i.e.}, it is $\lambda$-harmonic outside a small sphere (Section~\ref{sec:almost-harmonic}). For this activation function, we show at a stable critical point, we must learn at least one of the hidden nodes. Gradient descent (possibly with some noise, as in the work of Ge \emph{et al.}~\cite{GeHJY15}) is believed to converge to stable critical points. However, for simplicity, we descend along directions of negative curvature to escape saddle points. Our activation lacks some regularity conditions required in~\cite{GeHJY15}. We believe the results in \cite{jin2017escape} can be adapted to our setting to prove that perturbed gradient descent converges to stable critical points.






%

%
%

There is still a large gap between theory and practice. However, we believe our work can offer some theoretical explanations and guidelines for the design of better activation functions for gradient-based training algorithms. For example, better accuracy and training speed were reported when using the newly discovered exponential linear unit (ELU) activation function in \cite{ClevertUH15, ShahKSS16}. We hope for more theory-backed answers to these and many other questions in deep learning.


\subsection{Related Work.}
If the activation functions are linear or
if some independence assumptions are made, Kawaguchi shows that the
only local minima are the global minima \cite{Kawaguchi16a}. Under the
spin-glass and other physical models, some have shown that the loss
landscape admits well-behaving local minima that occur usually when the
overall error is small
\cite{ChoromanskaHMAL14, DauphinPGCGB14}. When only training
error is considered, some have shown that a global minima can be
achieved if the neural network contains sufficiently many hidden nodes
\cite{SoudryC16}. Recently, Daniely has shown that SGD learns the conjugate kernel class \cite{daniely2017sgd}. Under simplifying assumptions, some results for learning ReLU's with gradient descent are given in \cite{tian2017analytical, brutzkus2017globally}. Our research is inspired by
\cite{valiant2014learning}, where the authors show that for polynomial 
target functions, gradient descent on neural networks
with one hidden layer converges to low error, given a large
number of hidden nodes, and under complex perturbations,
there are no robust local minima. Even more recently, similar results about the convergence of SGD for two-layer neural networks have been established for a polynomial activation function under a more complex loss function \cite{ge2017learning}. And in \cite{li2017convergence}, they study the same problem as ours with the RELU activation and where lower layer of the network is close to identity and the upper layer has weights all one. This corresponds to the case where each electron is close to a distinct proton -- under these assumptions they show that SGD learns the true network. 

Under worst case assumptions, there has been hardness results for even simple networks. A neural network with one hidden unit and sigmoidal activation can admit exponentially many local minima \cite{Auer}. Backprogration has been proven to fail in a simple network due to the abundance of bad local minima \cite{brady1989back}. Training a 3-node neural network with one hidden layer is { NP}-complete \cite{BlumR88}.  But, these and many similar worst-case hardness results are based on worst case training data assumptions. However, by using a result in \cite{klivans2006cryptographic} that learning a neural network with threshold activation functions is equivalent to learning intersection of halfspaces, several authors showed that under certain cryptographic assumptions, depth-two neural networks are not efficiently learnable with smooth activation functions \cite{LivniSS14, ZhangLWJ15, ZhangLJ15}.


Due to the difficulty of analysis of the non convex gradient descent in deep learning, many have turned to improper learning and the study of non-gradient methods to train neural networks. Janzamin et. al use tensor decomposition methods to learn the shallow neural network weights, provided access to the score function of the training data distribution \cite{JanzaminSA15}. Eigenvector and tensor methods are also used to train shallow neural networks with quadratic activation functions in \cite{LivniSS14}. Combinatorial methods that exploit layerwise correlations in sparse networks have also been analyzed provably in \cite{AroraBGM13}. Kernel methods, ridge regression, and even boosting were explored for regularized neural networks with smooth activation functions in \cite{shalev2011learning, ZhangLWJ15, ZhangLJ15}. Non-smooth activation functions, such as the ReLU, can be approximated by polynomials and are also amenable to kernel methods\cite{GoelKKT16}. These methods however are very different from the simple popular SGD.

\section{Deep Learning, Potentials, and Electron-Proton Dynamics}
\label{sec:epdyn}
\subsection{Preliminaries.}
We will work in the space  $\mathcal{M}= \R^d.$ 
We denote the gradient and Hessian as $\nabla_{\R^d} f$ and  $\nabla_{\R^d}^2 f$ respectively.
The Laplacian is defined as
$\Delta_{\R^d} f = \Tr(\nabla_{\R^d}^2 f)$. 
If $f$ is multivariate with variable $x_i$, then let $f_{x_i}$ be a
restriction of $f$ onto the variable $x_i$ with all other variables
fixed. Let $\nabla_{x_i}f, \Delta_{x_i}f$ to be the gradient and
Laplacian, respectively, of $f_{x_i}$ with respect to $x_i$. Lastly,
we say $x$ is a critical point of $f$ if $\nabla f$ does not exist or
$\nabla f = 0$. 

We focus on learning depth two networks with a linear activation on
the output node. If the network takes inputs $x \in \R^d$ (say from
some distribution $\mathcal{D}$), then the network output, denoted
$f(x)$ is a sum over $k = \poly(d)$ hidden units with weight vectors
$w_{i} \in \R^d,$ activation $\sigma(x,w):\R^d \times \R^d\to \R,$ and
output weights $b_i \in \R.$ Thus, we can write
$f(x) = \sum_{i=1}^k b_i\sigma(x,w_i)$. We denote this concept class
$\mathcal{C}_{\sigma,k}.$ Our hypothesis concept class is also
$\mathcal{C}_{\sigma,k}.$ 





Let $\boldsymbol{a} = (a_1,...,a_k)$ and $\boldsymbol{\theta} =
(\theta_1,...,\theta_k)$; similarly for $\boldsymbol{b},
\boldsymbol{w}$ and our guess is $\hat{f}(x) = \sum_{i=1}^k a_i
\sigma(x, \theta_i)$. 
We define $\Phi,$ the {\bf potential function} {\it corresponding to the activation
  } $\sigma,$ as
\[\Phi(\theta, w) = \expt_{X\sim D}[ \sigma(X,\theta) \sigma(X,w)].\]
We work directly with the true squared loss error
$L(a,\theta)= \expt_{x \sim \mathcal{D}}[(f - \hat{f})^2]$. To
simplify $L$, we re-parametrize $a$ by $-a$ and expand.
%
%
\begin{align}
 L(\boldsymbol{a,\theta})  & = \expt_{X\sim D}\left[ \left(
  \sum_{i=1}^k a_i \sigma(X,\theta_i) + \sum_{i=1}^k
  b_i\sigma(X,w_i)\right)^2\right] \nonumber \\
& = \sum_{i=1}^k \sum_{j=1}^k a_i a_j \Phi(\theta_i,\theta_j)+ 2 a_ib_j \Phi(\theta_i,w_j)+ b_i b_j \Phi(w_i,w_j),
 \label{errLoss}
\end{align}
Given $\mathcal{D},$ the activation function $\sigma$, and the loss
$L$, we attempt to show that we can use some variant of gradient
descent to learn, with high probability, an $\epsilon$-approximation
of $w_j$ for some (or all) $j$. 
Note that our loss is jointly convex, though it is quadratic in
$\boldsymbol{a}$.

In this paper, we restrict our attention to translationally invariant activations
and potentials.
%
Specifically, we may write $\Phi= h(\theta-w)$ for some function $h(x).$ Furthermore, a translationally invariant function $\Phi(r)$ is {\it radial} if it is a function of $r = \|x - y \|$.

{\bf Remark}: Translationally symmetric potentials satisfy
$\Phi(\theta,\theta)$ {\it is a positive constant}. We normalize
$\Phi(\theta,\theta) = 1$ for the rest of the paper.
  
We assume that our input distribution
$\mathcal{D} = \mathcal{N}(0, {\bf I_{d\times d}})$ is fixed as the
standard Gaussian in $\R^d$. This assumption is not critical and a
simpler distribution might lead to better bounds. However, for
arbitrary distributions, there are hardness results for PAC-learning
halfspaces \cite{klivans2006cryptographic}.

We call a potential function {\bf realizable} if it corresponds to
some activation $\sigma$.  The following theorem characterizes
realizable translationally invariant potentials under standard
Gaussian inputs. Proofs and a similar characterization for rotationally invariant
potentials can be found in \supmaterial{Appendix~\ref{sec:realizable}}.
%
\begin{restatable}{theorem}{tranreal}
\label{thm:tranReal}
Let $\mathcal{M} = \R^d$ and $\Phi$ is square-integrable and
$\FT(\Phi)$ is integrable. Then, $\Phi$ is realizable under standard
Gaussian inputs if $\mathfrak{F}(\Phi)(\omega) \geq 0$ and the
corresponding activation is
$\sigma(x) =
(2\pi)^{d/4}e^{x^Tx/4}\mathfrak{F}^{-1}(\sqrt{\mathfrak{F}(\Phi)})(x),
$ where $\mathfrak{F}$ is the generalized Fourier transform in $\R^d.$
\end{restatable}
%

%
\subsection{Electron-Proton Dynamics}

%
%

By interpreting the pairwise potentials as electrostatic attraction
potentials, we notice that our dynamics is similar to electron-proton
type dynamics under potential $\Phi$, where $w_i$ are fixed point
charges in $\R^d$ and $\theta_i$ are moving point charges in $\R^d$
that are trying to find $w_i$. The total force on each charge is the
sum of the pairwise forces, determined by the gradient of $\Phi.$ We
note that standard dynamics interprets the force between particles as
an acceleration vector. In gradient descent, it is interpreted as a
velocity vector. 
\begin{definition}\label{EPDef}
  Given a potential $\Phi$ and particle locations
  $\theta_1,...,\theta_k \in \rea^d$ along with their respective
  charges $a_1,...,a_k \in \R$. We define {\bf Electron-Proton
    Dynamics} under $\Phi$ with some subset $S \subseteq [k]$ of fixed
  particles to be the solution $(\theta_1(t),...,\theta_k(t))$ to the
  following system of differential equations: For each pair
  $(\theta_i,\theta_j)$, there is a force from $\theta_j$ exerted on
  $\theta_i$ that is given by
  ${\bf F}_{i}(\theta_j) = a_ia_j\nabla_{\theta_i}
  \Phi(\theta_i,\theta_j)$ and
\[-\frac{d\theta_i}{dt} =  \sum_{ j \neq i} {\bf F}_{i}(\theta_j)\]
for all $i \not \in S$, with $\theta_i(0) = \theta_i$. For $i \in S$, $\theta_i(t) =  \theta_i$.
\end{definition}
For the following theorem, we assume that $\boldsymbol{\theta}$ is fixed.
\begin{restatable}{theorem}{epdyn}
\label{EPDyn}
Let $\Phi$ be a symmetric potential and $L$ be as in \eqref{errLoss}. Running continuous gradient descent on $\frac{1}{2} L$ with respect to $\theta$, initialized at
$(\theta_1,...,\theta_k)$ produces the same dynamics as
Electron-Proton Dynamics under $2\Phi$ with fixed particles at
$w_1,...,w_k$ with respective charges $b_1,..,b_k$ and moving
particles at $\theta_1,...,\theta_k$ with respective charges
$a_1,...,a_k$.
\end{restatable}

%

\section{Earnshaw's Theorem and Harmonic Potentials} 
\label{sec:earnshaw}
When running gradient descent on a non-convex loss, we often can and
do get stuck at a local minima. In this section, we use second-order
information to deduce that for certain classes of potentials, there
are no spurious local minima. The potentials In this section are often
{\it unbounded and un-realizable}. However, in the next section, we
apply insights developed here to derive similar convergence results
for approximations of these potentials.
%

%
Earnshaw's theorem in electrodynamics shows that there is no stable
local minima for electron-proton dynamics. This hinges on the property
that the electric potential
$\Phi(\theta,w) = \|\theta-w\|^{2-d}, d \neq 2$ is harmonic, with
$d = 3$ in natural setting. If $d = 2$, we instead have
$\Phi(\theta, w) = - \ln(\|\theta - w\|)$. First, we notice that this
is a symmetric loss, and our usual loss in \eqref{errLoss} has
constant terms that can be dropped to further simplify.
\begin{equation}\label{errSimp}
\overline{L}(a,\theta) =  2\sum_{i=1}^k\sum_{i < j} a_ia_j\Phi(\theta_i,\theta_j) + 2\sum_{i=1}^k\sum_{j=1}^ka_ib_j \Phi(\theta_i,w_j)
\end{equation} 
\begin{definition}
$\Phi(\theta,w)$ is a {\bf harmonic} potential on $\Omega$ if $\Delta_\theta \Phi(\theta,w) = 0$ for all $\theta \in \Omega$, except possibly at $\theta = w$.
\end{definition}

\begin{definition}
  Let $\Omega \subseteq \R^d$ and consider a function
  $f:\Omega \to \R$. A critical point $x^* \in \Omega$ is a {\bf local
    minimum} if there exists $\epsilon > 0$ such that
  $f(x^*+v) \geq f(x^*)$ for all $\|v\|\leq \epsilon$. It is a {\bf
    strict local minimum} if the inequality is strict for all
  $\norm{v} \le \epsilon.$
\end{definition} 
\begin{fact}
  Let $x^*$ be a critical point of a function $f : \Omega \to \R$ such
  that $f$ is twice differentiable at $x^*.$ Then, if $x^*$ is a local
  minimum then $\lambda_{min}(\nabla^2 f(x^*)) \geq 0.$ Moreover, if
  $\lambda_{min}(\nabla^2 f(x^*)) > 0,$ then $x^*$ is a strict local minimum.
\end{fact}

Note that if $\lambda_{min}(\nabla^2 f(x^*)) < 0$ then moving along the direction of the corresponding eigenvector decreases $f$ locally. If $\Phi$ is harmonic then it can be shown the trace of its Hessian is $0$ so 
if there is any non zero eigenvalue then at least one eigenvalue is negative. This idea results in the following known theorem (see full proof in supplementary material) that is applicable to the electric potential  
function $1/r$ in $3$-dimensions since is harmonic. It implies that a configuration of $n$ electrons and $n$ protons cannot be in a strict local minimum even if one of the mobile charges is isolated
(however note that this potential function goes to $\infty$ at $r=0$ and may not be realizable).

\begin{restatable}{theorem}{earnshaw}
\emph{(Earnshaw's Theorem. See~\cite{arnold1985mathematical})}
\label{Earnshaw} 
Let $\mathcal{M} = \R^d$ and let $\Phi$ be harmonic and $L$
be as in $\eqref{errSimp}$. Then, $L$ admits no
differentiable strict local minima.
\end{restatable}
Note that the Hessian of a harmonic potential can be identically
zero. To avoid this possibility we generalize harmonic potentials.
%

\subsection{\texorpdfstring{$\lambda$}{𝜆}-Harmonic Potentials}

%
In order to relate our loss function with its Laplacian, we consider potentials that are non-negative eigenfunctions of the Laplacian operator. Since the zero eigenvalue case simply gives rise to harmonic potentials, we restrict our attention to positive eigenfunctions.
\begin{definition}
A potential $\Phi$ is {\bf$\lambda$-harmonic} on $\Omega$ if there exists $\lambda > 0$ such that for every $\theta \in \Omega$, $\Delta_\theta \Phi(\theta, w) = \lambda \Phi(\theta,w) $, except possibly at $\theta = w$.
\end{definition}

Note that there are realizable versions of these potentials; for
example $\Phi(a,b) = e^{-\|a-b\|_1}$ in $\R^1.$
In the next section, we construct  realizable potentials that are 
$\lambda$-harmonic almost everywhere except when $\theta$ and $w$ are very close. 
\begin{restatable}{theorem}{eigStrict}
\label{EigStrict}
Let $\Phi$ be $\lambda$-harmonic and $L$ be as in \eqref{errLoss}. Then,
$L$ admits no local minima $\boldsymbol{(a,\theta)}$, except when
$L(\boldsymbol{a,\theta}) = L(0,\boldsymbol{\theta})$ or $\theta_i = w_j$ for some $i,j$. 
\end{restatable}
\begin{proof}
  Let $(\boldsymbol{a,\theta})$ be a critical point of $L.$ On the
  contrary, we assume that $\theta_i \neq w_j$ for all $i,j.$ WLOG, we
  can partition $[k]$ into $S_1,...,S_r$ such that for all $u \in S_i,
  v \in S_j$, we have $\theta_{u} = \theta_v$ iff $i=j$. 
Let $S_1 = \{ \theta_1, \ldots, \theta_{l}\}.$ 
We consider changing all
$\theta_1, \ldots, \theta_{l}$ by the same $v$ and define 
$H({\bf a}, v) = L({\bf a},\theta_1+v,...,\theta_l+v, \theta_{l+1}
\ldots, \theta_k).$

The optimality conditions on ${\bf a}$ are
$
0 = \pd{L}{a_i} =  2\sum_{j} a_j \Phi(\theta_i,\theta_j)
  + 2\sum_{j=1}^k b_j \Phi(\theta_i,w_j).
$
Thus, by the definition of $\lambda$-harmonic potentials, we may differentiate as $\theta_i \neq w_j$ and compute the Laplacian as 
\begin{align*}
\Delta_v H & = \lambda\sum_{i=1}^l a_i \left(2\sum_{j=1}^k b_j
  \Phi(\theta_i, w_j) + 2\sum_{j=l+1}^k a_j
  \Phi(\theta_i, \theta_j)\right) \\
& = \lambda\sum_{i=1}^l a_i \left( - 2
  \sum_{j = 1}^l  a_j \Phi(\theta_i,\theta_j)\right) 
%
%
 = -2\lambda\sum_{i=1}^l a_i \left( 
  \sum_{j = 1}^l  a_j \right) = -2 \lambda\left(\sum_{i=1}^l a_i\right)^2
\end{align*} 
If $\sum_{i=1}^l a_i \neq 0$, then we conclude that the Laplacian is strictly negative, so we are not at a local minimum. Similarly, we can conclude that for each $S_i,$ $\sum_{u \in S_i} a_u = 0$. In this case, since $\sum_{i=1}^k a_i \sigma(\theta_i,x) = 0$, $L(\boldsymbol{a,\theta}) = L(0,\boldsymbol{\theta})$. 
\end{proof} 
%
%


\section{Realizable Potentials with Convergence Guarantees}
\label{sec:almost-harmonic}
In this section, we derive convergence guarantees for realizable
potentials that are almost $\lambda$-harmonic, specifically, they are
$\lambda$-harmonic outside of a small neighborhood around the
origin. First, we prove the existence of activation functions such
that the corresponding potentials are almost $\lambda$-harmonic. Then,
we reason about the Laplacian of our loss, as in the previous section,
to derive our guarantees. We show that at a stable minima, each of the
$\theta_i$ is close to some $w_j$ in the target network. We may end up
with a many to one mapping of the learned hidden weights to the true
hidden weights, instead of a bijection.
To make sure that $\|a\|$ remains controlled throughout the
optimization process, we add a quadratic regularization term to $L$
and instead optimize $G = L + \|a\|^2$.


%

Our optimization procedure is a slightly altered version of gradient
descent, where we incorportate a second-order method (which we call
Hessian descent as in Algorithm~\ref{HD}) that is used when the
gradient is small and progress is slow. The descent algorithm
(Algorithm~\ref{SecondGD}) allows us to converge to points with small
gradient and small negative curvature. Namely, for smooth functions,
in $\poly(1/\epsilon)$ iterations, we reach a point in
$\mathcal{M}_{G, \epsilon}$, where
\begin{align*}
\mathcal{M}_{G, \epsilon} = \left\{x\in \mathcal{M} \Big| \|\nabla G(x)\|
  \leq \epsilon \text{ and } \lambda_{min}(\nabla^2 G(x)) \geq
  -\epsilon\right\}
  \end{align*}
We show that if $(\boldsymbol{a,\theta})$ is in $\mathcal{M}_{G, \epsilon}$ for $\epsilon$ small, then $\theta_i$ is close to $w_j$ for some $j$. Finally, we show how to initialize $(\boldsymbol{a^{(0)},\theta^{(0)}})$ and run second-order GD to converge to $\mathcal{M}_{G,\epsilon}$, proving our main theorem.

%
%
\alglanguage{pseudocode}
\begin{algorithm}[hb]
 \caption{$x = HD(L,x_0, T,\alpha$)}
   \label{HD}
\begin{algorithmic}
   \State {\bfseries Input:} $L: \mathcal{M} \to \R$; $x_0 \in \mathcal{M}$; $T\in \N$; $\alpha \in \R$
   \State Initialize $x \leftarrow x_0$
   \For {$i=1$ {\bfseries to} $T$}
   \State Find unit eigenvector $v_{min}$ corresponding to $\lambda_{min}(\nabla^2 f(x))$ 
   \State  $\beta \leftarrow -\alpha \lambda_{min}(\nabla^2 f(x)) \sign(\nabla f(x)^Tv_{min}) $
    \State $x \leftarrow x + \beta v_{min}$
   \EndFor
\end{algorithmic}
\end{algorithm}
\begin{algorithm}[hb]
 \caption{$x = SecondGD(L, x_0, T,\alpha, \eta, \gamma)$}
   \label{SecondGD}
\begin{algorithmic}
   \State {\bfseries Input:} $L:\mathcal{M} \to \R$; $x_0 \in
   \mathcal{M}$; $T\in \N$; $\alpha, \eta, \gamma \in \R$
   \For {$i=1$ {\bfseries to} $T$}
   \If {$\|\nabla L(x_{i-1})\| \geq  \eta$}\  $x_{i} \leftarrow
 x_{i-1} - \alpha \nabla L(x_{i-1})$
   \Else \
$x_i \leftarrow HD(L, x_{i-1}, 1, \alpha)$ 
  \EndIf
   \If {$ L(x_i) \geq L(x_{i-1}) - min(\alpha\eta^2/2, \alpha^2 \gamma^3/2) $}
     \Return $x_{i-1}$
   \EndIf 
   \EndFor
   \end{algorithmic}
\end{algorithm}
\begin{theorem}\label{almostHarmSGD}
  Let $\mathcal{M} = \R^{d}$ for $d \equiv 3 \mod 4$ and
  $k = \poly(d)$. For all $\epsilon \in (0,1),$ we can construct an
  activation $\sigma_\epsilon$ such that if $w_1,...,w_k \in \R^d$
  with $w_i$ randomly chosen from
  $w_i \sim \mathcal{N}({\bf 0}, O(d\log d){\bf I_{d\times d}})$ and
  $b_1,...,b_k$ be randomly chosen at uniform from $[-1,1]$, then with
  high probability, we can choose an initial point
  $(\boldsymbol{a^{(0)}, \theta^{(0)}})$ such that after running
  SecondGD (Algorithm \ref{SecondGD}) on the regularized objective
  $G(\boldsymbol{a,\theta})$ for at most $(d/\epsilon)^{O(d)}$
  iterations, there exists an $i, j$ such that
  $\|\theta_i - w_j\| < \epsilon$.
\end{theorem}
We start by stating a lemma concerning the construction of an almost
$\lambda$-harmonic function on $\R^d.$ The construction is given in \supmaterial{Appendix~\ref{sec:realizable}} and uses a linear combination of realizable potentials that correspond to an activation function of the indicator function of a $n$-sphere. By using Fourier analysis and Theorem~\ref{thm:tranReal}, we can finish the construction of our almost $\lambda$-harmonic potential.%
\begin{restatable}{lemma}{almostharmreal}\label{almostHarmReal}
  Let $\mathcal{M} = \R^d$ for $d \equiv 3 \mod 4$. Then, for any
  $\epsilon \in (0,1)$, we can construct a radial activation
  $\sigma_\epsilon(r)$ such that the corresponding radial potential
  $\Phi_\epsilon(r)$  is $\lambda$-harmonic for
  $r \geq \epsilon$.

Furthermore, we have ${\Phi_\epsilon}^{(d-1)}(r) \geq 0$ for all $r  > 0,$  ${\Phi_\epsilon}^{(k)}(r) \geq 0,$ and ${\Phi_\epsilon}^{(k+1)}(r)\leq 0$ for all $r > 0$ and $d - 3 \geq k \geq 0 $ even. 

When $\lambda = 1$, $|{\Phi}_\epsilon^{(k)}(r)| \leq O((d/\epsilon)^{2d})$ for all $0 \leq k \leq d-1$. And when $r \geq \epsilon$, $\Omega(e^{-r}r^{2-d}(d/\epsilon)^{-2d}) \leq {\Phi}_\epsilon(r) \leq O((1+r)^de^{1-r}(r)^{2-d})$ and $ \Omega(e^{-r}r^{1-d}(d/\epsilon)^{-2d}) \leq |{\Phi}_\epsilon'(r)| \leq O((d+r)(1+ r)^de^{1- r} r^{1-d})$
\end{restatable}
Our next lemma use the almost $\lambda$-harmonic properties to show that at an almost stationary point of $G$, we must have converged close to some $w_j$ as long as our charges $a_i$ are not too small. The proof is similar to Theorem \ref{EigStrict}. Then, the following lemma relates the magnitude of the charges $a_i$ to the progress made in the objective function. 
%
\begin{restatable}{lemma}{almostharmconv}\label{almostHarmConv}
  Let $\mathcal{M} = \R^d$ for $d \equiv 3 \mod 4$ and let $G$ be the
  regularized loss corresponding to the activation 
  $\sigma_\epsilon$ given by Lemma~\ref{almostHarmReal} with
  $\lambda =1$. For any $\epsilon \in (0,1)$ and $\delta \in (0, 1)$,
  if $\boldsymbol{(a,\theta)} \in \mathcal{M}_{G,\delta}$, then for
  all $i$, either 1) there exists $j$ such that
  $\|\theta_i - w_j\| < k\epsilon$ or 2) $a_i^2 < 2kd\delta$.
\end{restatable}
\begin{restatable}{lemma}{almostharmres}\label{almostHarmRes}
  Assume the conditions of Lemma~\ref{almostHarmConv}. If
$\sqrt{G({\bf a, \boldsymbol{\theta}})} \leq \sqrt{G(\boldsymbol{0,0})} - \delta$
  and $(\boldsymbol{a,\theta}) \in \mathcal{M}_{G,\delta^2/(2k^3d)}$,
  then there exists some $i, j$ such that $\|\theta_i - w_j\| <k\epsilon$.
\end{restatable}
 Finally, we guarantee that our initialization substantially decreases our objective function. Together with our previous lemmas, it will imply that we must be close to some $w_j$ upon convergence. This is the overview of the proof of Theorem~\ref{almostHarmSGD}, presented below.
 \begin{restatable}{lemma}{almostharminitialize}\label{almostHarmInitialize}
Assume the conditions of Theorem~\ref{almostHarmSGD} and Lemma~\ref{almostHarmConv}. With high probability, we can initialize $\boldsymbol{(a^{(0)},\theta^{(0)})}$ such that $\sqrt{G({\bf a^{(0)}}, \boldsymbol{\theta^{(0)}})} \leq \sqrt{G(\boldsymbol{0,0})} -\delta$ with $\delta = (d/\epsilon)^{ - O(d)}$.
 \end{restatable}
\begin{proof}[Proof of Theorem \ref{almostHarmSGD}]
  Let our potential $\Phi_{\epsilon/k}$ be the one as constructed in Lemma~\ref{almostHarmReal} that is $1$-harmonic for all $r \geq \epsilon/k$ and as always, $k = \poly(d)$.  First, by Lemma~\ref{almostHarmInitialize}, we can initialize $\boldsymbol{(a^{(0)},\theta^{(0)})}$ such that $\sqrt{G(\boldsymbol{a^{(0)},\theta^{(0)}})} \leq \sqrt{G({\bf 0,0})} - \delta$ for $\delta = (d/\epsilon)^{-O(d)}$. If we set $\alpha = (d/\epsilon)^{-O(d)}$ and $\eta = \gamma = \delta^2/(2k^3d)$, then running Algorithm~\ref{SecondGD} will terminate and return some $(\boldsymbol{a,\theta})$ in at most $(d/\epsilon)^{O(d)}$ iterations. This is because our algorithm ensures that our objective function decreases by at least $\min(\alpha \eta^2/2, \alpha^2\gamma^3/2)$ at each iteration, $G({\bf 0, 0})$ is bounded by $O(k),$ and $G \geq 0$ is non-negative.

Let $\boldsymbol{\theta} = (\theta_1,...\theta_k)$. If there exists $\theta_i, w_j$ such that $\|\theta_i - w_j\| < \epsilon$, then we are done. Otherwise, we claim that $(\boldsymbol{a,\theta}) \in \mathcal{M}_{G, \delta^2/(2k^3d)}$. For the sake of contradiction, assume otherwise. By our algorithm termination conditions, then it must be that after one step of gradient or Hessian descent from $(\boldsymbol{a,\theta})$, we reach some $(\boldsymbol{a',\theta'})$ and $G(\boldsymbol{a',\theta'}) > G(\boldsymbol{a,\theta}) - \min(\alpha\eta^2/2,\alpha^2\gamma^3/2)$.

Now, Lemma~\ref{almostHarmReal} ensures all first three derivatives of
$\Phi_{\epsilon/k}$ are bounded by $O((dk/\epsilon)^{2d})$, except at
$w_1,...,w_k$. Furthermore, since there do not exist
$\theta_i, w_j$ such that $\|\theta_i - w_j\| <\epsilon$, $G$ is
three-times continuously differentiable within a
$\alpha (dk/\epsilon)^{2d} = (d/\epsilon)^{-O(d)}$ neighborhood of
$\boldsymbol{\theta}$. Therefore, by Lemma~\ref{GradDecrease} and
~\ref{HessianDecrease} in the \supmaterial{appendix}, we must have
$G(\boldsymbol{a',\theta'}) \leq G(\boldsymbol{a,\theta}) -
\min(\alpha\eta^2/2,\alpha^2\gamma^3/2),$ a contradiction. Lastly,
since our algorithm maintains that our objective function is
decreasing, so
$\sqrt{G(\boldsymbol{a,\theta})} \leq \sqrt{G({\bf 0,0})} -
\delta$. Finally, we conclude by Lemma \ref{almostHarmRes}.
\end{proof}

\subsection{Node-by-Node Analysis}
We cannot easily analyze the convergence of gradient descent to the global minima when all $\theta_i$ are simultaneously moving since the pairwise interaction terms between the $\theta_i$ present complications, even with added regularization. Instead, we run a greedy
node-wise descent (Algorithm~\ref{NodeGDOpt}) to learn the hidden weights, i.e. we run a descent algorithm with respect to $(a_i,\theta_i)$ sequentially. The
main idea is that after running SGD with respect to $\theta_1$,
$\theta_1$ should be close to some $w_j$ for some $j$. Then, we can
carefully induct and show that $\theta_2$ must be some other $w_k$ for
$k\neq j$ and so on.

\begin{algorithm}[tb]
 \caption{Node-wise Descent Algorithm}
   \label{NodeGDOpt}
\begin{algorithmic}
  \State {\bfseries Input:}
  $(\boldsymbol{a,\theta}) = (a_1,...,a_k,\theta_1,...,\theta_k), a_i
  \in\R, \theta_i\in\mathcal{M}$;
  $T\in \N$; $L$; $\alpha, \eta, \gamma \in \R$; 
  \For {$i=1$ {\bfseries to} $k$} 
  \State{\bf Initialize} $(a_i, \theta_i)$
  \State $(a_i, \theta_i) = SecondGD \left(L_{a_i, \theta_i},(a_i,\theta_i),T, \alpha,\eta,\gamma \right)$
   \EndFor
   \State {\bf return} $a = (a_1,...,a_k), \theta = (\theta_1,..., \theta_k)$
   \end{algorithmic}
\end{algorithm}

Let $L_1(a_1,\theta_1)$ be the objective $L$ restricted to $a_1,\theta_1$ being variable, and $a_2,...,a_k = 0$ are fixed. The tighter control on the movements of $\theta_1$ allows us to remove our regularization. While our previous guarantees before allow us to reach a $\epsilon$-neighborhood of $w_j$ when running SGD on $L_1$, we will strengthen our guarantees to reach a $(d/\epsilon)^{-O(d)}$-neighborhood of $w_j$, by reasoning about the first derivatives of our potential in an $\epsilon$-neighborhood of $w_j$. By similar argumentation as before, we will be able to derive the following convergence guarantees for node-wise training. 

\begin{restatable}{theorem}{nodewise}\label{nodeWise}
Let $\mathcal{M} = \R^{d}$ and $d \equiv 3 \mod 4$ and let $L$ be as in \ref{errLoss} and $k = \poly(d)$. For all $\epsilon \in (0,1),$ we can construct an activation $\sigma_\epsilon$ such that if $w_1,...,w_k \in \R^d$ with $w_i$ randomly chosen from $w_i \sim  \mathcal{N}({\bf 0}, O(d\log d){\bf I_{d\times d}})$ and $b_1,...,b_k$ be randomly chosen at uniform from $[-1,1]$, then with high probability, after running nodewise descent (Algorithm \ref{NodeGDOpt}) on the objective $L$ for at most $(d/\epsilon)^{O(d)}$ iterations, $\boldsymbol{(a,\theta)}$ is in a $(d/\epsilon)^{-O(d)}$ neighborhood of the global minima.
\end{restatable}

%


\section{Experiments}
\label{experiments}
For our experiments, our training data is given by $(x_i, f(x_i))$, where $x_i$ are randomly chosen from a standard Gaussian in $\R^d$ and $f$ is a randomly generated neural network with weights chosen from a standard Gaussian. We run gradient descent (Algorithm \ref{GD}) on the empirical loss, with stepsize around $\alpha = 10^{-5}$, for $T = 10^6$ iterations. The nonlinearity used at each node is sigmoid from -1 to 1, including the output node, unlike the assumptions in the theoretical analysis. A random guess for the network will result in a mean squared error of around 1. Our experiments (see Fig~\ref{expconverge}) show that for depth-2 neural networks, even with non-linear outputs, the training error diminishes quickly to under $0.002$. This seems to hold even when the width, the number of hidden nodes, is substantially increased (even up to 125 nodes), but depth is held constant; although as the number of nodes increases, the rate of decrease is slower. This substantiates our claim that depth-2 neural networks are learnable.

However, it seems that for depth greater than 2, the test error becomes significant  when width is high (see Fig~\ref{tablePlot}). Even for depth 3 networks, the increase in depth impedes the learnability of the neural network and the training error does not get close enough to 0. It seems that for neural networks with greater depth, positive convergence results in practice are elusive. We note that we are using training error as a measure of success, so it's possible that the true underlying parameters are not learned. 

\begin{figure}[!ht]
  \centering
\includegraphics[width = 4.5in]{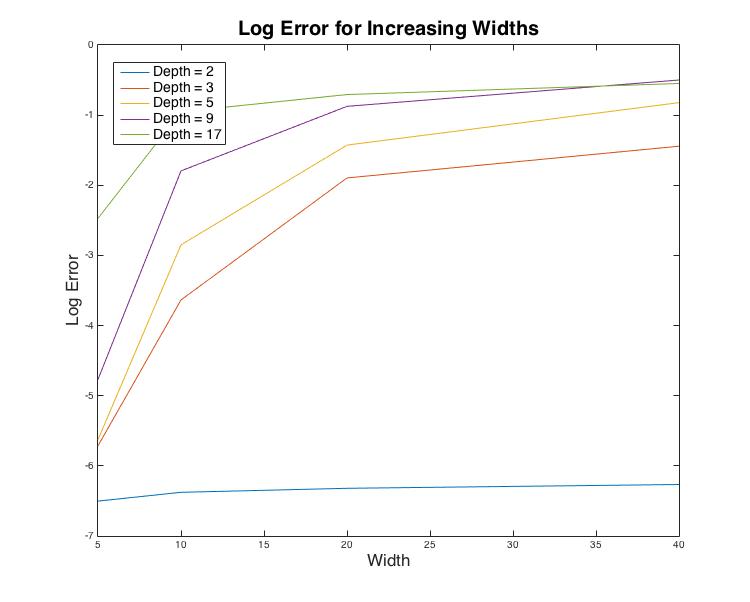}
\caption{Test Error of Varying-Depth Networks vs. Width}
\label{tablePlot}
\end{figure}

\begin{table}
\vskip 0.1in
\begin{center}
\begin{small}
\begin{sc}
\begin{tabular}{
  |p{\dimexpr.2\linewidth-2\tabcolsep-1.3333\arrayrulewidth}|
   |p{\dimexpr.2\linewidth-2\tabcolsep-1.3333\arrayrulewidth}
  |p{\dimexpr.2\linewidth-2\tabcolsep-1.3333\arrayrulewidth}
  |p{\dimexpr.2\linewidth-2\tabcolsep-1.3333\arrayrulewidth}
  |p{\dimexpr.2\linewidth-2\tabcolsep-1.3333\arrayrulewidth}
   |p{\dimexpr.2\linewidth-2\tabcolsep-1.3333\arrayrulewidth}
  }
   \hline 
           & Width 5   &  Width 10   & Width 20 & Width 40     \\ \hline 
    Depth 2 & 0.0015   & 0.0017      &   0.0018 & 0.0019 \\ \hline
    Depth 3 & 0.0033   & 0.0264        &   0.1503 & 0.2362 \\ \hline
    Depth 5 & 0.0036   & 0.0579        &   0.2400 & 0.4397 \\ \hline
    Depth 9 & 0.0085   & 0.1662        &   0.4171 & 0.6071 \\ \hline
    Depth 17 & 0.0845   & 0.3862        &   0.4934 & 0.5777 \\ \hline
\end{tabular}
\end{sc}
\end{small}
\end{center}
\caption{Test Error of Learning Neural Networks of Various Depth and Width}
\vskip -0.1in
\end{table}

{\small
\bibliography{biblio}
\bibliographystyle{alpha}}

\newpage
\appendix
\section{Electron-Proton Dynamics}

\epdyn*

\begin{proof}
The initial values are the same. Notice that continuous gradient descent on $L(\boldsymbol{a,\theta})$ with respect to $\theta$ produces dynamics given by $\frac{d\theta_i(t)}{dt} = -\nabla_{\theta_i}L(\boldsymbol{a,\theta})$. Therefore,
\[\frac{d\theta_i(t)}{dt} = -2\sum_{j \neq i} a_i a_j
\nabla_{\theta_i}\Phi(\theta_i,\theta_j) - 2\sum_{j=1}^k
a_ib_j\nabla_{\theta_i} \Phi(\theta_i,w_j)\] 
And gradient descent does not move $w_i$. By definition, the dynamics corresponds to Electron-Proton Dynamics as claimed.
\end{proof}

\section{Realizable Potentials}
\label{sec:realizable}

\subsection{Activation-Potential Calculations}
First define the {\it dual} of a function $f: \R \to \R$ is defined to be 
\[ \widehat{f}(\rho) = \expt_{X,Y \sim N(\rho)}[f(X)f(Y)],\]
where $N(\rho)$ is the bivariate normal distribution with $X, Y$ unit variance and $\rho$ covariance. This is as in \cite{DanielyFS16}.
\begin{lemma}\label{rotLem}
Let $\mathcal{M} = S^{d-1}$ and $\sigma$ be our activation function, then $\widehat{\sigma}$ is the corresponding potential function.
\end{lemma}

\begin{proof}
If $u, v$ have norm 1 and if $X$ is a standard Gaussian in $\R^d$, then note that $X_1 = u^TX$ and $X_2 = v^TX$ are both standard Gaussian variables in $\R^1$ and the covariance is $E[X_1X_2] = u^Tv$. 

Therefore, the dual function of the activation gives us the potential function.
\begin{align*}
\expt_{X}[\sigma(u^TX)\sigma(v^TX)] & =
\expt_{X,Y \sim N(u^Tv)}[\sigma(X)\sigma(Y)] \\
& = \widehat{\sigma}(u^Tv).
\end{align*}
\end{proof}

By Lemma \ref{rotLem}, the calculations of the activation-potential
for the sign, ReLU, Hermite, exponential functions are given in
\cite{DanielyFS16}. For the Gaussian and Bessel activation functions,
we can calculate directly. In both case, we notice that we may write
the integral as a product of integrals in each dimension. Therefore,
it suffices to check the following 1-dimensional identities.
\begin{align*}
  & \int_{-\infty}^\infty
    \sqrt{2}e^{x^2/4}e^{-(x-\theta)^2}\sqrt{2}e^{x^2/4}e^{-(x-w)^2} \frac{1}{\sqrt{2\pi}} e^{-x^2/2}\, dx \\
  & \qquad = \sqrt{\frac{2}{\pi}}\int_{-\infty}^\infty
    e^{-(x-\theta)^2}e^{-(x-w)^2} \, dx = e^{-(\theta -w)^2/2}
\end{align*}
\begin{align*}
& \int_{-\infty}^\infty (\frac{2}{\pi})^{3/2}e^{x^2/2}K_0(|x-\theta|)K_0(|x-w|)  \frac{1}{\sqrt{2\pi}} e^{-x^2/2}\, dx \\
& \qquad 
= \int_{-\infty}^\infty \frac{2}{\pi^2}K_0(|x-\theta|)K_0(|x-w|) \, dx
  = e^{-|\theta -w|}
\end{align*}

The last equality follows by Fourier uniqueness and taking the Fourier transform of both sides, which are both equality $\sqrt{2/\pi}(\omega^2+1)^{-1}$.

\subsection{Characterization Theorems}

\tranreal*

\begin{proof}
Since $\Phi$ is square-integrable, its Fourier transform exists. Let $h(x) = \FT^{-1}(\sqrt{\FT(\Phi)})(x)$ and this is well-defined since the Fourier transform was non-negative everywhere and the Fourier inverse exists since $\sqrt{\FT(\Phi)}(x)$ is square-integrable. Now, let $\sigma(x,w) = (2\pi)^{1/4}e^{x^2/4}h(x-w)$. Realizability follows by the Fourier inversion theorem:
\begin{align*}
    \expt_{X \sim N}[\sigma(X,w)\sigma(X,\theta)]  &= \int_{\R^n} h(x-w)h(x-\theta) \, dx \\
    &= \int_{\R^n} h(x)h(x-(\theta-w)) \, dx \\
    &= \FT^{-1}(\FT(h\ast h)(\theta -w)) \\
    &= \FT^{-1}(\FT(h)^2(\theta - w)) \\
    &= \FT^{-1}(\FT(\Phi)(\theta - w)) \\
    &= \Phi(\theta - w) 
\end{align*}
 
Note that $\ast$ denotes function convolution.
\end{proof}

When our relevant space is $\mathcal{M} = S^{d-1},$ we let
$\Pi_\mathcal{M}$ be the projection operator on $\mathcal{M}$. The
simplest way to define the gradident on
$S^{d-1}$ is $\nabla_{S^{d-1}} f(x) = \nabla_{\R^d} f(x/\|x\|)$, where
$\| \cdot \|$ denotes the $l_2$ norm and $x \in S^{d-1}$. 
The Hessian and Laplacian are analogously defined and the subscripts
are usually dropped where clear from context.

We say that a potential $\Phi$ on $\mathcal{M} = S^{d-1}$ is
rotationally invariant if for all $\theta, w \in S^{d-1},$ we have
 $\Phi = h(\theta^Tw).$
\begin{restatable}{theorem}{rotreal}
\label{thm:rotReal}
Let $\mathcal{M} = S^{d-1}$ and $\Phi(\theta,w) = f(\theta^Tw)$. Then,
$\Phi$ is realizable if $f$ has non-negative Taylor coefficients, $c_i
\geq 0$ , and the corresponding activation $\sigma(x) = \sum_{i=1}^\infty \sqrt{c_i} h_i(x)$
converges almost everywhere, where $h_i(x)$ is the i-th Hermite polynomial.
\end{restatable}

\begin{proof}
By \ref{rotLem} and due to the orthogonality of hermite polynomials, if $f = \sum_i a_i h_i$, where $h_i(x)$ is the i-th Hermite polynomial, then
\[\widehat{f}(\rho) = \sum_{i} a_i^2 \rho^i\]

Therefore, any function with non-negative taylor coefficients is a valid potential function, with the corresponding activation function determined by the sum of hermite polynomials, and the sum is bounded almost everywhere by assumption.
\end{proof}

\subsection{Further Characterizations}

To apply Theorem~\ref{thm:tranReal}, we need to check that the Fourier transform of our function is non-negative. Not only is this is not straightforward to check, many of our desired potentials do not satisfy this criterion. In this section, we would like to have a stronger characterization of realizable potentials, allowing us to construct realizable potentials that approximates our desired potential.
 
\begin{definition}
Let $\Phi$ be a positive semidefinite function if for all $x_1,...,x_n$, the matrix $A_{ij} = \Phi(x_i - x_j)$ is positive semidefinite. 
\end{definition}

\begin{lemma}\label{lem:psd}
Let $\mathcal{M} = \R^d$ and $\Phi(\theta, w) = f(\theta-w)$ is is realizable, then it is positive semidefinite.
\end{lemma}

\begin{proof}
If $\Phi$ is realizable, then there exists $\sigma$ such that  $\Phi(\theta, w) =  \expt_{X \sim N}[\sigma(X, w)\sigma(X, \theta)]$. For $x_1,...,x_n$, we note that the quadratic form:
\begin{align*}
\sum_{i,j} \Phi(x_i,x_j) v_i v_j 
& = \sum_{i,j} \expt_{X \sim N}[\sigma(X, x_i)\sigma(X, x_j)] v_i v_j 
 = \expt_{X \sim N}\left[\left(\sum_{i}v_i \sigma(X , x_i) \right)^2\right] \geq 0
\end{align*}

Since $\Phi$ is translationally symmetric, we conclude that $\Phi$ is positive semidefinite.
\end{proof}

\begin{definition}
A potential $\Phi$ is $\FT$-integrable if it is square-integrable and $\FT({\Phi}(\omega))$ is integrable, where $\FT$ is the standard Fourier transform.
\end{definition}

\begin{lemma}\label{intReal}
Let $w(x) \geq 0$ be a positive weighting function such that $\int_a^b w(x) \, dx$ is bounded. If $\Phi_x$ is a parametrized family of $\FT$-integrable realizable potentials, then, $\int_a^b w(x) \Phi_x$ is $\FT$-integrable realizable.
\end{lemma}

\begin{proof}
Let $\Phi = \int_a^b w(x) \Phi_x$. From linearity of the Fourier transform and $\int_a^b w(x)\, dx$ is bounded, we know that $\Phi$ is $\FT$-integrable. Since $\Phi_x$ are realizable, they are positive definite by Lemma~\ref{lem:psd} and by Bochner's theorem, their Fourier transforms are non-negative. And since $w(x) \geq 0$, we conclude by linearity and continuity of the Fourier transform that $\FT(\Phi) \geq 0$. By Theorem \ref{thm:tranReal}, we conclude that $\Phi$ is realizable.
\end{proof}

\begin{lemma}\label{baseConstruct}
Let $\mathcal{M} = \R^d$ for $d \equiv 3 \mod 4$. Then, for any $\epsilon, t > 0$, there exists a $\FT$-integrable realizable $\Phi$ such that for $t \geq r > \epsilon$, $\Phi^{(d-1)}(r) = t -r$ and for $ r \leq \epsilon$, $\Phi^{(d-1)}(r) = \frac{t-\epsilon}{\epsilon}r$. Furthermore, $\Phi^{(k)}(r) = 0$ for $r > t$ for all $0 \leq k \leq d$.
\end{lemma}

\begin{proof}
Our construction is based on the radial activation function $h_t(x, \theta) = \bf{1}_{\|\theta - x\| \leq t/2}$, which is the indicator in the disk of radius $t/2$. This function, when re-weighted correctly as $\sigma_t(x,\theta) =  (2\pi)^{1/4} e^{x^2/4}h_t(x,\theta)$ gives rise to a radial potential function that is simply the convolution of $h_t$ with itself, measuring the volume of the intersection of two spheres of radius $t$ centered at $\theta$ and $w$.
\begin{align*}
\Phi_t(\theta, w) & = \expt_X[\sigma_t(X,\theta)\sigma_t(X,w)]
= \begin{cases}
C\int_{\|\theta - w\|/2}^{t/2} ((t/2)^2 - x^2)^{(d-1)/2} \, dx & \|\theta - w\| \leq t\\
0 & \textrm{ otherwise.}
\end{cases}
\end{align*}
Therefore, as a function of $r=\|\theta - w \|$, we see that when $r \leq t$, $\Phi_t(r) = C\int_{r/2}^{t/2} ((t/2)^2-x^2)^{(d-1)/2} \, dx$ and $\Phi_t'(r) = -C'((t/2)^2-(r/2)^2)^{(d-1)/2}$. Since $d \equiv 3 \mod 4$, we notice that $\Phi_t'$ has a positive coefficient in the leading $r^{d-1}$ term and since it is a function of $r^2$, it has a zero $r^{d-2}$ term. Therefore, we can scale $\Phi_t$ such that 
\begin{align*}
\Phi_t^{(d-1)}(r) = \begin{cases}
r & r \leq t\\
0 & \textrm{ otherwise.}
\end{cases} 
\end{align*}

$\Phi_t$ is clearly realizable and now we claim that it is $\FT$-integrable. First, $\Phi_t$ is bounded on a compact set so it is square-integrable. Now, since $\Phi_t = h_t \ast h_t$ can be written as a convolution, $\FT(\Phi_t) = \FT(h_t)^2$. Since $h_t$ is square integrable, then by Parseval's, $\FT(h_t)$ is square integrable, allowing us to conclude that $\Phi_t$ is $\FT$-integrable.

Now, for any $\epsilon > 0$, let us construct our desired $\Phi$ by taking a positive sum of $\Phi_t$ and then appealing to Lemma \ref{intReal}. Consider
\[\Phi(r) = \int_{\epsilon}^{t} \frac{1}{x^2}\Phi_x(r) \, dx\]

First, note that the total weight $\int_\epsilon^t \frac{1}{x^2}$ is bounded. Then, when $r \geq t$, since $\Phi_x(r) = 0$ for $x \leq t$, we conclude that $\Phi^{(k)}(r) = 0$ for any $k$. Otherwise, for $\epsilon < r < t$, we can apply dominated convergence theorem to get
\begin{align*}
\Phi^{(d-1)}(r) = \int_{\epsilon}^r \frac{1}{x^2}\Phi_x^{(d-1)}(r) \, dx + \int_{r}^t \frac{1}{x^2} \Phi_x^{(d-1)}(r) \, dx
 = 0 + \int_r^t \frac{r}{x^2} \, dx = 1 -r/t 
\end{align*}

Scaling by $t$ gives our desired claim. For $r\leq \epsilon$, we integrate similarly and scale by $t$ to conclude.
\end{proof}

\begin{lemma}\label{transConstruct}
Let $\mathcal{M} = \R^d$ for $d \equiv 3 \mod 4$ and let $\Phi(r)$ be a radial potential. Also, $\Phi^{(k)}(r) \geq 0$ and $\Phi^{(k+1)}(r)\leq 0$ for all $r > 0$ and $k \geq 0 $ even, and $\lim_{r \to \infty} \Phi^{(k)}(r) = 0$ for all $0 \leq k \leq d$. 

Then, for any $\epsilon > 0$, there exists a $\FT$-integrable realizable potential $\overline{\Phi}$  such that $\overline{\Phi}^{(k)}(r) = \Phi^{(k)}(r)$ for all $0 \leq k \leq d-1$ and $r \geq \epsilon$. Furthermore, we have $\overline{\Phi}^{(d-1)}(r) \geq 0$ for all $r  > 0$ and $\overline{\Phi}^{(k)}(r) \geq 0$ and $\overline{\Phi}^{(k+1)}(r)\leq 0$ for all $r > 0$ and $d - 3 \geq k \geq 0 $ even.

Lastly, for $r < \epsilon$ and $0 \leq k \leq d-1$, $|\overline{\Phi}^{(d-1-k)}(r)| \leq |\Phi^{(d-1-k)}(\epsilon)| + \sum_{j=1}^k \frac{(\epsilon - r)^{k-j+1}}{(k-j+1)!} |\Phi^{(d-j)}(\epsilon)|$
\end{lemma}

\begin{proof}
By Lemma \ref{baseConstruct}, we can find $\Phi_t$ such that
\begin{align*}
\Phi_t^{(d-1)} = \begin{cases}
\frac{t-\epsilon}{\epsilon } r & 0 \leq r \leq \epsilon \\
t - r & \epsilon < r \leq t \\
0 & r > t
\end{cases}
\end{align*}

Furthermore, $\Phi_t^{(k)}(r) = 0$ for $r > t$ for all $0 \leq k \leq d$.
Therefore, we consider 
\[\overline{\Phi}(r) = \int_{\epsilon}^\infty \Phi^{(d+1)}(x) \Phi_x(r) \, dx\] 

Note that this is a positive sum with $\int_{\epsilon}^\infty \Phi^{(d+1)}(x) \, dx = -\Phi^{(d)}(\epsilon) < \infty$. By the non-negativity of our summands, we can apply dominated convergence theorem and Fubini's theorem to get
\begin{align*}
\overline{\Phi}^{(d-1)}(r) & = \int_{\epsilon}^\infty  \Phi^{(d+1)}(x) (\Phi_x^{(d-1)}(r)) \, dx \\
& = \int_{r}^\infty  \Phi^{(d+1)}(x) (\Phi_x^{(d-1)}(r)) \, dx \\
& = \int_r^\infty \Phi^{(d+1)}(x) \int_r^x 1 \, dy \,dx  \\
& = \int_r^\infty \int_y^\infty \Phi^{(d+1)}(x) \, dx \, dy = \int_r^\infty  -\Phi^{(d)}(y) \, dy \\
& = \Phi^{(d-1)}(r)
\end{align*}

Now, since $\overline{\Phi}^{(d-1)}(r) = \Phi^{(d-1)}(r)$ for $r\geq \epsilon$ and $\lim_{r\to\infty} \overline{\Phi}^{(k)}(r) = \lim_{r\to\infty} {\Phi}^{(k)}(r) = 0$ for $0 \leq k \leq d-1$, repeated integration gives us our claim. 

Finally, for the second claim, notice that for $r \leq \epsilon$, we get
\begin{align*}
\overline{\Phi}^{(d-1)}(r) = \int_{\epsilon}^\infty  \Phi^{(d+1)}(x) (\Phi_x^{(d-1)}(r)) \, dx
=  r \int_\epsilon^\infty \Phi^{(d+1)}(x) \frac{x - \epsilon}{\epsilon} \, dx = Cr  
\end{align*}

Note that our constant $C \geq 0$ since the summands are non-negative. Therefore, we conclude that $\overline{\Phi}^{(d-1)}(r) \geq 0$ for all $r > 0$. Repeated integration and noting that $\lim_{r\to\infty} \overline{\Phi}^{(k)}(r) = 0$ for $0 \leq k \leq d-1$ gives us our claim.

Lastly, we prove the last claim of the theorem with induction on $k$. This holds trivially for $k = 0$ since $\overline{\Phi}^{(d-1)}(r) \leq \overline{\Phi}^{(d-1)}(\epsilon) = \Phi^{(d-1)}(\epsilon)$ for $r \leq \epsilon$. Then, assume we have the inequality for $k < d-1$. By integration, we have
\begin{align*}
|\overline{\Phi}^{(d-k-2)}(r)| & \leq |\overline{\Phi}^{(d-k-2)}(\epsilon)| + \int_r^\epsilon |\overline{\Phi}^{(d-1 -k)}(y)| \, dy \\
& \leq |\overline{\Phi}^{(d-k-2)}(\epsilon)| + \int_r^\epsilon |\overline{\Phi}^{(d-1 -k)}(\epsilon)| \, dy \\
& + \int_r^\epsilon \sum_{j=1}^k \frac{(\epsilon -y)^{k-j+1}}{(k-j+1)!}|\Phi^{(d-j)}(\epsilon)| \, dy \\
& \leq  |\overline{\Phi}^{(d-k-2)}(\epsilon)| + \sum_{j=1}^{k+1}  \frac{(\epsilon -y)^{k-j+2}}{(k-j+2)!}|\Phi^{(d-j)}(\epsilon)|
\end{align*}

Therefore, we conclude with induction.
\end{proof}

\almostharmreal*

\begin{proof}
This is a special case of the following lemma.
\end{proof}

\begin{lemma}
Let $\mathcal{M} = \R^d$ for $d \equiv 3 \mod 4$. Then, for any $1 > \epsilon > 0$, we can construct a radial activation $\sigma_\epsilon(r)$ with corresponding normalized radial potential $\Phi_\epsilon(r)$ that is $\lambda$-harmonic when $r \geq \epsilon$.

Furthermore, we have ${\Phi_\epsilon}^{(d-1)}(r) \geq 0$ for all $r  > 0$ and ${\Phi_\epsilon}^{(k)}(r) \geq 0$ and ${\Phi_\epsilon}^{(k+1)}(r)\leq 0$ for all $r > 0$ and $d - 3 \geq k \geq 0 $ even. 

Also, $|{\Phi}_\epsilon^{(k)}(r)| \leq 3(2d + \sqrt{\lambda})^{2d} \epsilon^{-2d}e^{\sqrt{\lambda}}$ for all $0 \leq k \leq d-1$. And for $r \geq \epsilon$, $e^{-\sqrt{\lambda}r}r^{2-d}(2d+\sqrt{\lambda})^{-2d}\epsilon^{2d}/3\leq {\Phi}_\epsilon(r) \leq (1+r\sqrt{\lambda})^de^{\sqrt{\lambda}(1-r)}(r)^{2-d}$. Also for $r \geq \epsilon$, $ e^{-\sqrt{\lambda}r}r^{1-d}(2d+\sqrt{\lambda})^{-2d}\epsilon^{2d}/3 \leq |{\Phi}_\epsilon'(r)| \leq (d+\sqrt{\lambda}r)(1+ r\sqrt{\lambda})^de^{\sqrt{\lambda}(1- r)} r^{1-d}$
\end{lemma}

\begin{proof}
Consider a potential of the form $\Phi(r) = p(r)e^{-\sqrt{\lambda}r}/r^{d-2}$. We claim that there exists a polynomial $p$ of degree $k = (d-3)/2$ with non-negative coefficients and $p(0) = 1$ such that $\Phi$ is $\lambda$-harmonic. Furthermore, we will also show along the way that $p(r) \leq (1+\sqrt{\lambda}r)^d$.

When $d = 3$, it is easy to check that $\Phi(r) = e^{(-\sqrt{\lambda})r}/r$ is our desired potential. Otherwise, by our formula for the radial Laplacian in $d$ dimensions, we want to solve the following differential equation:
\[\Delta \Phi =  \frac{1}{r^{d-1}} \frac{\partial}{\partial r} (r^{d-1} \frac{\partial \Phi}{\partial r}) =\lambda \Phi\]

Solving this gives us the following second-order differential equation on $p$
\[rp'' - (d-3+2\sqrt{\lambda}r)p' +\sqrt{\lambda} (d-3)p = 0\]

Let us write $p(r) = \sum_{i=0}^k a_i r^i$. Then, substituting into our differential equation gives us the following equations by setting each coefficient of $r^i$ to zero:

$r^i$:  $a_{i+1}(i+1)(i - (d-3)) = a_i \sqrt{\lambda} (2i-(d-3))$

$r^k:$ $(-2k +d-3)a_k = 0$

The last equation explains why we chose $k = (d-3)/2$, so that it is automatically zero. Thus, setting $a_0 = 1$ and running the recurrence gives us our desired polynomial. Note that the recurrence is valid and produces positive coefficients since $i < k  = (d-3)/2$. Our claim follows and $\Phi$ is $\lambda$-harmonic. And furthermore, notice that $a_{i+1} \leq \sqrt{\lambda} a_i \leq (\sqrt{\lambda})^{i+1}$. Therefore, $p(r) \leq (1+r\sqrt{\lambda})^d$. 

Lastly, we assert that $\Phi^{(j)}(r)$ is non-negative for $j$ even and non-positive for $j$ odd. To prove our assertion, we note that it suffices to show that if $\Phi$ is of the form $\Phi(r) = p(r) e^{-\sqrt{\lambda}r}/r^{l}$ for some $p$ of degree $k < l$ and $p$ has non-negative coefficients, then $\Phi'(r) = - q(r) e^{-\sqrt{\lambda}r}/r^{l+1}$ for some $q$ of degree $k+1$ with non-negative coefficients. 

Differentiating $\Phi$ gives:
\[\Phi' = \frac{e^{-r}}{r^{l+1}} (rp'(r) - (l + \sqrt{\lambda} r)p(r))\]

It is clear that if $p$ has degree $k$, then $q(r) = (l+\sqrt{\lambda} r)p(r) - rp'(r)$ has degree $k+1$, so it suffices to show that it has non-negative coefficients. Let $p_0,..., p_k$ be the non-negative coefficients of $p$. Then, by our formula, we see that 

$q_0 = l p_0$

$q_i = lp_i - ip_i + \sqrt{\lambda}p_{i-1} = (l-i)p_i + \sqrt{\lambda}p_{i-1}$ 

$q_{k+1} = \sqrt{\lambda} p_k$

Since $i \leq k < l$, we conclude that $q$ has non-negative coefficients. Finally, our assertion follows with induction since $\Phi^{(0)}(r)$ is non-negative and has our desired form with $k = (d-3)/2 < d-2$. By Lemma \ref{transConstruct}, our primary theorem follows, we can construct a realizable radial potential $\Phi_\epsilon(r)$ that is $\lambda$-harmonic when $r \geq \epsilon$ and has alternating-signed derivatives.

Lastly, we prove the following preliminary bound on $\Phi_\epsilon^{(k)}(r)$ when $k \leq d$: $|\Phi_\epsilon^{(k)}(r)| \leq 3(2d + \epsilon \sqrt{\lambda})^{2d}\epsilon^{-2d}  $ for all $0 \leq k \leq d-1$. First, notice that by the results of Lemma~\ref{transConstruct}, $\Phi_\epsilon^{(k)}(r)$ is monotone and $\lim_{r\to\infty}\Phi_\epsilon^{(k)}(r) = 0$. So, it follows that we just have to bound $|\Phi_{\epsilon}^{(k)}(0)|$. From our construction, $\Phi_{\epsilon}^{(k)}(\epsilon) = p_k(\epsilon)e^{-\sqrt{\lambda}\epsilon}\epsilon^{2-d-k}$, for some polynomial $p_k$. Furthermore, from our construction, we have the recurrence $p_{k}(\epsilon) = (d-2+k + \sqrt{\lambda}\epsilon)p_{k-1}(\epsilon) - \epsilon p_{k-1}'(\epsilon)$. Therefore, we conclude that for $k \leq d$, $p_k(\epsilon) \leq (2d + \sqrt{\lambda} \epsilon)^kp_0(\epsilon) \leq  (2d + \sqrt{\lambda} \epsilon)^k(1+\sqrt{\lambda}\epsilon)^d \leq (2d + \sqrt{\lambda}\epsilon)^{2d}$. 

Therefore, we can bound
 $|\Phi_\epsilon^{(k)}(\epsilon)| \leq (2d + \sqrt{\lambda}\epsilon)^{2d}\epsilon^{-2d}$. Finally, by Lemma~\ref{transConstruct}, 
\begin{align*}
|\Phi_{\epsilon}^{(d-1-k)}(0)| 
& \leq |\Phi_{\epsilon}^{(d-1-k)}(\epsilon)| + \sum_{j=1}^k \frac{(\epsilon )^{k-j+1}}{(k-j+1)!} |\Phi^{(d-j)}(\epsilon)| \\
& \leq (2d + \sqrt{\lambda}\epsilon)^{2d}\epsilon^{-2d} (1 + \sum_{j=1}^k \frac{\epsilon^{k-j+1}}{(k-j+1)!}) \\
& \leq (2d + \sqrt{\lambda}\epsilon)^{2d}\epsilon^{-2d} e^{\epsilon} \leq 3(2d + \sqrt{\lambda}\epsilon)^{2d}\epsilon^{-2d} 
\end{align*}

And for $r \geq \epsilon$, we see that $|\Phi_\epsilon(r)| = |\Phi(r)| \leq |p(r)|\frac{e^{-\sqrt{\lambda}r}}{r^{d-2}} = (1+r\sqrt{\lambda})^de^{-\sqrt{\lambda}r}r^{2-d}$. And $|\Phi_\epsilon'(r)| = |\Phi'(r)| \leq |p_1(r)| \frac{e^{-\sqrt{\lambda} r}}{r^{d-1}} \leq (d+\sqrt{\lambda}r)(1+ r\sqrt{\lambda})^de^{-\sqrt{\lambda} r} r^{1-d}$.

Finally, we consider the normalized potential: $\widetilde{\Phi}_\epsilon = {\Phi}_\epsilon/{\Phi}_\epsilon(0)$. Note that since ${\Phi}_\epsilon$ is monotonically decreasing, we can lower bound ${\Phi}_\epsilon(0) \geq {\Phi}_\epsilon(\epsilon) \geq e^{-\sqrt{\lambda}}$. Therefore, we can derive the following upper bounds: $|\widetilde{\Phi}_\epsilon^{(k)}(r)| \leq 3(2d + \sqrt{\lambda})^{2d} \epsilon^{-2d}e^{\sqrt{\lambda}}$ and for $r \geq \epsilon$, $|\widetilde{\Phi}_\epsilon(r)| \leq (1+r\sqrt{\lambda})^de^{\sqrt{\lambda}(1-r)}r^{2-d}$ and its derivative is bounded by $|\widetilde{\Phi}_\epsilon'(r)| \leq (d+\sqrt{\lambda}r)(1+ r\sqrt{\lambda})^de^{\sqrt{\lambda}(1- r)} r^{1-d}$.

And lastly, we derive some lower bounds on $\widetilde{\Phi}_\epsilon$ and the first derivative when $r \geq \epsilon$, by using the upper bound on $\Phi_\epsilon(0)$: $\widetilde{\Phi}_\epsilon(r) \geq {\Phi}_\epsilon(r)(2d+\sqrt{\lambda})^{-2d}\epsilon^{2d}/3  \geq e^{-\sqrt{\lambda}r}r^{2-d}(2d+\sqrt{\lambda})^{-2d}\epsilon^{2d}/3$. For the derivative, we get
 $|\widetilde{\Phi}_\epsilon'(r)| \geq e^{-\sqrt{\lambda}r}r^{1-d}(2d+\sqrt{\lambda})^{-2d}\epsilon^{2d}/3$.
\end{proof}

\begin{lemma}\label{3dlambdaharmonic}
The $\lambda$-harmonic radial potential $\Phi(r) = e^{-r}/r$ in $3$-dimensions is realizable by the activation 
$\sigma(r) = K_1(r)/r$.
\end{lemma}

\begin{proof}
The activation is obtained from the potential function by first taking its Fourier transform, then taking its square root, and then taking the inverse fourier transform. Since the functions in consideration are radially symmetric the Fourier transform $F(y)$ of $f(x)$ (and inverse) are obtained by the Hankel Transfom $y F(y) = \int_0^\infty x f(x) J_{1/2} (xy) \sqrt{xy} dx$. Plugging $f(x) = e^{-x}/x$, from the Hankel tranform tables we get $yF(y) = cy/(1+y^2)$ giving $F(y) = cy/(1+y^2)$. So we wish to find the inverse Fourier transform for $1/\sqrt{1+y^2}$. The inverse $f(x)$ is given by $xf(x) = \int_0^\infty y F(y) J_{1/2} (xy) \sqrt{xy} dy = cK_1(x)$. So $\sigma(r) = K_1(r)/r$.
\end{proof}

\section{Earnshaw's Theorem}

\earnshaw*
\begin{proof}
  If $(\boldsymbol{a,\theta})$ is a differentiable strict local
  minima, then for any $i,$ we must have
\[\nabla_{\theta_{i}} L = 0, \text{ and }  \Tr(\nabla^2_{\theta_i}L) > 0.\]
Since $\Phi$ is harmonic, we also have
\begin{align*}
&  \Tr(\nabla^2_{\theta_i}L(\theta_1,...,\theta_n)) = \Delta_{\theta_i} L 
  =  2\sum_{ j\neq i} a_ia_j \Delta_{\theta_i}\Phi(\theta_i,\theta_j)
  + 2\sum_{j=1}^ka_ib_j  \Delta_{\theta_i}\Phi(\theta_i,w_j) = 0,
\end{align*}
which is a contradiction. In the first line, there is a factor of 2 by symmetry.
\end{proof}

\section{Descent Lemmas and Iteration Bounds}
\begin{algorithm}[hb]
 \caption{$x = GD(L,x_0, T,\alpha$)}
   \label{GD}
\begin{algorithmic}
   \State {\bfseries Input:} $L: \mathcal{M} \to \R$; $x_0 \in \mathcal{M}$; $T\in \N$; $\alpha\in \R$
   \State Initialize $x = x_0$
   \For {$i=1$ {\bfseries to} $T$}
   \State $x = x - \alpha\nabla L(x)$
   \State $x = \Pi_\mathcal{M} x$
   \EndFor
\end{algorithmic}
\end{algorithm}

\begin{lemma}\label{GradDecrease}		
Let $f :\Omega \to \R$ be a thrice differentiable function such that $|f(y)| \leq B_0, \|\nabla f(y)\| \leq B_1, \|\nabla^2 f(y)\| \leq B_2,\|\nabla^2f (z)-\nabla^2L(y)\| \leq B_3\|z - y\|$ for all $y,z$ in a $(\alpha B_1)$-neighborhood of $x$. If $\|\nabla f (x) \|\geq \eta$ and $x'$ is reached after one iteration of gradient descent (Algorithm \ref{GD}) with stepsize $\alpha \leq \frac{1}{B_2}$, then $\|x' - x\| \leq \alpha B_1$ and $f(x') \leq f(x) - \alpha\eta^2/2$.
\end{lemma} 

\begin{proof}
The gradient descent step is given by $x' = x - \alpha \nabla f(x)$. The bound on $\|x' - x\|$ is clear since $\|\nabla f(x) \| \leq B_1$.
\begin{align*}
f(x') &\leq f(x) - \alpha \nabla f(x)^T\nabla f(x)^T + \alpha^2\frac{B_2}{2} \|\nabla f(x)\|^2 \\
&\leq f(x) - (\alpha - \alpha^2 \frac{B_2}{2}) \eta^2 
\end{align*}
For $0 \leq \alpha \leq \frac{1}{B_2}$, we have $\alpha - \alpha^2B_2/2 \geq \alpha/2$, and our lemma follows.
\end{proof}

\begin{lemma}\label{HessianDecrease}
Let $f :\Omega \to \R$ be a thrice differentiable function such that $|f(y)| \leq B_0, \|\nabla f(y)\| \leq B_1, \|\nabla^2 f(y)\| \leq B_2,\|\nabla^2f (z)-\nabla^2L(y)\| \leq B_3\|z - y\|$ for all $y,z$ in a $(\alpha B_2)$-neighborhood of $x$. If $\lambda_{min}(\nabla^2 f (x)) \leq -\gamma$ and $x'$ is reached after one iteration of Hessian descent (Algorithm \ref{HD}) with stepsize $\alpha \leq \frac{1}{B_3}$, then $\|x' - x\| \leq \alpha B_2$ and $f(x') \leq f(x) - \alpha^2 \gamma^3/2$.
\end{lemma}

\begin{proof}
The gradient descent step is given by $x' = x + \beta v_{min}$, where $v_{min}$ is the unit eigenvector corresponding to $\lambda_{min}(\nabla^2f(x))$ and $\beta = -\alpha\lambda_{min}(\nabla^2 f(x))sgn(\nabla f(x)^Tv_{min})$. Our bound on $\|x' - x\|$ is clear since $|\lambda_{min}(\nabla^2 f(x))| \leq B_2$.
\begin{align*}
f(x') &\leq f(x) + \beta\nabla f(x)^Tv_{min} + \beta^2 v_{min}^T\nabla^2f(x)v_{min} + \frac{B_3}{6} |\beta|^3 \|v_{min}\|^3 \\
&\leq f(x) - |\beta|^2 \gamma + \frac{B_3}{6} |\beta|^3
\end{align*}
The last inequality holds since the sign of $\beta$ is chosen so that $\beta \nabla f(x)^Tv_{min} \leq 0$. Now, since $|\beta| = \alpha \gamma \leq \frac{\gamma}{B_3}$, $-|\beta|^2\gamma + \frac{B_3}{6} |\beta|^3 \leq - \alpha^2 \gamma^3/2$. 
\end{proof}

\section{Convergence of Almost $\lambda$-Harmonic Potentials}\label{App:EigenFunc}

\almostharmconv*

\begin{proof}
 The proof is similar to Theorem \ref{EigStrict}. Let $\Phi_\epsilon$ be the realizable potential in \ref{almostHarmReal} such that $\Phi_\epsilon(r)$ is $\lambda$-harmonic when $r \geq \epsilon$ with $\lambda = 1$. Note that $\Phi_\epsilon(0) = 1$ is normalized. And let $\boldsymbol{(a,\theta)} \in \mathcal{M}_{G,\delta}$. 
 
WLOG, consider $\theta_1$ and a initial set $S_0 = \{ \theta_1\}$ containing it. For a finite set of points $S$ and a point $x$, define $d(x,S) = \min_{y \in S} \| x - y\|$. Then, we consider the following set growing process. If there exists $\theta_i, w_i \not \in S_j$ such that $d(\theta_i, S_j) < \epsilon$ or $d(w_i, S_j) < \epsilon$, add $\theta_i, w_i$ to $S_j$ to form $S_{j+1}$. Otherwise, we stop the process. We grow $S_0$ to until the process terminates and we have the grown set $S$.

If there is some $w_j \in S$, then it must be the case that there exists ${j_1},\cdots {j_q}$ such that $\|\theta_1 - \theta_{j_1} \| < \epsilon$ and
$\|\theta_{j_{i}} - \theta_{j_{i+1}}\| < \epsilon$, and
$\|\theta_{j_q}- w_j\| <\epsilon$ for some $w_j$. So, there exists $j$, such that $\|\theta_1 - w_j\| < k\epsilon$. 

Otherwise, notice that for each $\theta_i \in S$, $\|w_j - \theta_i\|\geq \epsilon$ for all $j$, and $\|\theta_i - \theta_j\| \geq \epsilon$ for all $\theta_j\not \in S$. WLOG, let $S = \{\theta_1,\dots,\theta_l\}$. 
  
We consider changing all
$\theta_1, \ldots, \theta_{l}$ by the same $v$ and define 
\[H({\bf a}, v) = G({\bf a},\theta_1+v,...,\theta_l+v, \theta_{l+1}
\ldots, \theta_k).\]

The optimality conditions on ${\bf a}$ are 
\begin{align*}
   \abs{\pd{H}{a_i}} & = \lvert 4a_i  + 2\sum_{j\neq i} a_j \Phi_\epsilon(\theta_i,\theta_j) + 2\sum_{j=1}^k b_j \Phi_\epsilon(\theta_i,w_j) \rvert \leq \delta
\end{align*}
Next, since $\Phi_\epsilon(r)$ is $\lambda$-harmonic for $r \geq \epsilon$, we may calculate the Laplacian of $H$ as
\begin{align*}
\Delta_v H & = \sum_{i=1}^l \lambda \left(2\sum_{j=1}^k a_ib_j
  \Phi_\epsilon(\theta_i, w_j) + 2\sum_{j=l+1}^k a_ia_j
  \Phi_\epsilon(\theta_i, \theta_j)\right) \\
& \leq \sum_{i=1}^l \lambda \left(-4 a_i^2 - 2
  \sum_{j = 1, j\neq i}^l  a_ia_j \Phi_\epsilon(\theta_i,\theta_j)\right)+ \delta \sum_{i=1}^l \lambda |a_i| \\
&= -2\lambda\expt\left[\left( \sum_{i=1}^l a_i \sigma(\theta_i,X)\right)^2\right] -2\lambda \sum_{i=1}^l a_i^2+ \delta\lambda\sum_{i=1}^l  |a_i| 
\end{align*} 
The second line follows from our optimality conditions and the third line follows from completing the square. Since $\boldsymbol{(a,\theta)} \in \mathcal{M}_{G,\delta}$, we have $\Delta_v H \geq - 2kd\delta$. Let $S = \sum_{i=1}^l a_i^2$. Then, by Cauchy-Schwarz, we have $-2 \lambda S + \delta\lambda\sqrt{k} \sqrt{S} \geq -2kd\delta$. When $S \geq \delta^2 k$, we see that $-\lambda S \geq -2 \lambda S + \delta\lambda \sqrt{k}\sqrt{S} \geq -2kd\delta$. Therefore, $S \leq 2kd\delta/\lambda$.
 
We conclude that $S \leq \max(\delta^2k, 2kd\delta/\lambda) \leq 2kd\delta/\lambda$ since $\delta\leq 1 \leq 2d/\lambda$ and $\lambda = 1$. Therefore, $a_i^2 \leq 2kd\delta$.
\end{proof}

\almostharmres*

 \begin{proof}
 If there does not exists $i, j$ such that
   $\|\theta_i - w_j\| <k\epsilon$, then by Lemma \ref{almostHarmConv}, this implies $a_i^2 < \delta^2/k^2$ for all $i$. Now, for a integrable
   function $f(x)$, $\| f\|_X = \sqrt{\expt_X[f(X)^2]}$ is a
   norm. Therefore, if $f(x) = \sum_i b_i \sigma(w_i,x)$ be our true
   target function, we conclude that by triangle inequality
\begin{align*}
\sqrt{G(\boldsymbol{a,\theta})}  & \geq \norm{\sum_{i=1}^k a_i \sigma(\theta_i,x) - f(x)}_X \geq \|f(x)\|_X\ - \sum_{i=1}^k \|a_i\sigma(\theta_i,x) \|_X \geq
  \sqrt{G(\boldsymbol{0,0})} - \delta
\end{align*}
This gives a contradiction, so we conclude that there must exist $i, j$ such that $\theta_i$ is in a $k\epsilon$ neighborhood of $w_j$.
 \end{proof}
 
\almostharminitialize*

 \begin{proof}
  Consider choosing $\theta_1 = {\bf 0}$ and then
  optimizing $a_1$. Given $\theta_1$, the loss decrease is:
\begin{align*}
   G(a_1,{\bf 0}) - G({\bf 0},{\bf 0}) & = \min_{a_1} 2a_1^2 +
  2\sum_{j=1}^k a_1 b_j\Phi_\epsilon({\bf 0},w_j) 
 = -\frac{1}{2}\left(  \sum_{j=1}^k b_j
   \Phi_\epsilon({\bf 0},w_j)\right)^2 
\end{align*}

Because $w_j$ are random Gaussians with variance $O(d \log d)$, we have $\|w_j\| \leq O(d\log d)$ with high probability for all $j$. By Lemma~\ref{almostHarmReal}, our potential satisfies ${\Phi}_\epsilon({\bf 0}, w_j) \geq (d/\epsilon)^{ - O(d)}$. And since $b_j$ are uniformly chosen in $[-1,1]$, we conclude that with high probability over the choices of $b_j$, $-\frac{1}{2}\left( \sum_{j=1}^k b_j\Phi(\theta_1,w_j)\right)^2 \geq (d/\epsilon)^{ - O(d)}$ by appealing to Chebyshev's inequality on the squared term. 


Therefore, we conclude that with high probability, $G(a_1, {\bf 0}) \leq G(\boldsymbol{0,0}) - \frac{1}{2}(d/\epsilon)^{ - O(d)}$. Let $\sqrt{G(a_1, {\bf 0})} = \sqrt{G(\boldsymbol{0,0})} - \Delta \geq 0$. Squaring and rearranging gives $\Delta \geq \frac{1}{4\sqrt{G({\bf 0, 0})}}(d/\epsilon)^{ - O(d)}$. Since $G(\boldsymbol{0,0}) \leq O(k) = O(\poly(d))$, we are done. 
\end{proof}
%
%

\subsection{Node by Node Analysis}

The first few lemmas are similar to the ones proven before in the simultaneous case. The proof are presented for completeness because the regularization terms are removed. Note that our loss function is quadratic in $\bf{ a}$. Therefore, let $a_1^*(\theta_1)$ denote the optimal value of $a_1$ to minimize our loss. 

\begin{lemma}\label{nodeConv}
Let $\mathcal{M} = \R^d$ for $d \equiv 3 \mod 4$ and let $L_1$ be the loss restricted to $(a_1,\theta_1)$ corresponding to the activation function $\sigma_\epsilon$ given by Lemma~\ref{almostHarmReal} with $\lambda = 1$. For any $\epsilon \in (0,1)$ and $\delta \in (0, 1)$, we can construct $\sigma_\epsilon$ such that if ${(a_1,\theta_1)} \in \mathcal{M}_{L_1,\delta}$, then for all $i$, either 1) there exists $j$ such that $\|\theta_1 - w_j\| < \epsilon$ or 2) $a_1^2 < 2d\delta$.
\end{lemma}

\begin{proof}
The proof is similar to Lemma~\ref{almostHarmConv}. Let $\Phi_\epsilon$ be the realizable potential in \ref{almostHarmReal} such that $\Phi_\epsilon(r)$ is $\lambda$-harmonic when $r \geq \epsilon$. Note that $\Phi_\epsilon(0) = 1$ is normalized. And let $(a_1,\theta_1) \in \mathcal{M}_{L_1,\delta}$. Assume that there does exist $w_j$ such that $\|\theta_1 - w_j\| < \epsilon$. 
 
The optimality condition on ${ a_1}$ is
\begin{align*}
   \abs{\pd{L}{a_1}} & = \lvert 2a_1  + 2\sum_{j=1}^k b_j \Phi_\epsilon(\theta_1,w_j) \rvert \leq \delta
\end{align*}
Next, since $\Phi_\epsilon(r)$ is $\lambda$-harmonic for $r \geq \epsilon$, we may calculate the Laplacian of $L$ as
\begin{align*}
\Delta_{\theta_1} L & = \lambda \left(2\sum_{j=1}^k a_1b_j
  \Phi_\epsilon(\theta_1, w_j) \right) 
 \leq  -2\lambda a_1^2 + \delta \lambda |a_1| 
\end{align*} 
The inequality follows from our optimality conditions. Since ${(a_1,\theta_1)} \in \mathcal{M}_{L,\delta}$, we have $\Delta_{\theta_1} L \geq - 2d\delta$. When $a_1^2 \geq \delta^2$, we see that $-\lambda a_1^2 \geq -2 \lambda a_1^2 + \delta\lambda |a_1| \geq -2d\delta$. Therefore, $a_1^2 \leq 2d\delta/\lambda$. We conclude that $a_1^2 \leq \max(\delta^2, 2d\delta/\lambda) \leq 2d\delta/\lambda$ for $\delta\leq 2d \leq 2d/\lambda$ since $\lambda = 1$. Therefore, $a_1^2 \leq 2d\delta$.
\end{proof}
\begin{lemma}\label{nodeRes}
Assume the conditions of Lemma~\ref{nodeConv}. If
$\sqrt{L_1(a_1,\theta_1)} \leq \sqrt{L_1(0, 0)} - \delta$
  and $(a_1,\theta_1) \in \mathcal{M}_{G,\delta^2/(2d)}$,
  then there exists some $j$ such that $\|\theta_1 - w_j\| <\epsilon$.
\end{lemma}
\begin{proof}
The proof follows similarly from Lemma \ref{almostHarmRes}.
\end{proof}
 Now, our main observation is below, showing that in a neighborhood around $w_j$, descending along the gradient direction will move $\theta_1$ closer to $w_j$. Our tighter control of the gradient of $\Phi_{\epsilon}$ around $w_j$ will eventually allow us to show that $\theta_1$ converges to a small neighborhood around $w_j$. 
 
\begin{lemma}\label{nodeGradient}
Assume the conditions of Theorem~\ref{nodewiseSGD} and Lemma~\ref{nodeConv}. If $\|\theta_1 - w_j\| \leq d$ and $|b_j|\geq 1/\poly(d)$ and $|a_1 - a_1^*(\theta_1)| \leq (d/\epsilon)^{-O(d)}$ is almost optimal and for $i$, $\|w_i - w_j\| \geq \Omega(d \log d)$, then $-\nabla_{\theta_1}L_1 = \zeta \frac{w_j - \theta_1}{\|\theta_1 - w_j\|} + \xi$ with $\zeta \geq  \frac{1}{\poly(d)}(d/\epsilon)^{-8d}$ and $\xi \leq (d/\epsilon)^{-O(d)}$. 
\end{lemma}

\begin{proof}
Through the proof, we assume $k = \poly(d)$. Now, our gradient with respect to $\theta_1$ is
\begin{align*}
\nabla_{\theta_1} L_1 &= 2a_1b_j \Phi_\epsilon'(\|\theta_1 - w_j\|) \frac{\theta_1 - w_j}{\|\theta_1 - w_j\|}+ 2\sum_{i\neq j} a_1b_i\Phi_\epsilon'(\|\theta_1 - w_i\|) \frac{\theta_1 - w_i}{\|\theta_1 - w_i\|}
\end{align*}

Since $\|\theta_1 - w_j\| \leq d$, we may lower bound $|\Phi_\epsilon'(\|\theta_1 - w_j\|)| \geq e^{-\sqrt{\lambda}d}d^{1-d}(2d+\sqrt{\lambda})^{-2d}\epsilon^{2d}/3 \geq \Omega((d/\epsilon)^{-4d})$. Similarly, $\Phi_\epsilon(\|\theta_1 - w_j\|) \geq \Omega((d/\epsilon)^{-4d})$. On the other hand since $\|w_i - w_j\| \geq \Omega(d \log d)$ for all $i \neq j$, we may upper bound $|\Phi_{\epsilon}(\|\theta_1 - w_i\|)| \leq (d/\epsilon)^{-O(d)}$ and $|\Phi_{\epsilon}'(\|\theta_1 - w_i\|)| \leq (d/\epsilon)^{-O(d)}$. Together, we conclude that $\nabla_{\theta_1} L_1 = 2a_1b_j \Phi_\epsilon'(\|\theta_1 - w_j\|) \frac{\theta_1 - w_j}{\|\theta_1 - w_j\|} + 2a_1\xi$, where $\|\xi\| \leq (d/\epsilon)^{-O(d)}$.

By assumption, $|a_1 - a_1^*(\theta_1)| \leq (d/\epsilon)^{-O(d)}$, so 
\begin{align*}
   \abs{\pd{L_1}{a_1}} & = \lvert 2a_1 + 2b_j \Phi_\epsilon(\|\theta_1 - w_j\|) + 2\sum_{i \neq j} b_i \Phi_\epsilon(\|\theta_1 - w_i\|) \rvert \leq (d/\epsilon)^{-O(d)}
\end{align*}

By a similar argument as on the derivative, we see that $a_1 = -b_j \Phi_\epsilon(\|\theta_1 - w_j\|) + (d/\epsilon)^{-O(d)}$. Therefore, the direction of $-\nabla_{\theta_1} L_1$ is moving $\theta_1$ closer to $w_j$ since 
\begin{align*}
-\nabla_{\theta_1} L_1 &=  b_j^2\Phi_\epsilon(\|\theta_1-w_j\|)\Phi_\epsilon'(\|\theta_1 - w_j\|) \frac{\theta_1 - w_j}{\|\theta_1 - w_j\|} + (d/\epsilon)^{-O(d)} 
\end{align*}

and we know $\Phi_\epsilon > 0$ and $\Phi_\epsilon' < 0$, thereby $-b_j^2\Phi_\epsilon(\|\theta_1-w_j\|)\Phi_\epsilon'(\|\theta_1 - w_j\|) \geq 1/\poly(d)(d/\epsilon)^{-8d}$.
\end{proof}

 \begin{lemma}[Node-wise Initialization]\label{nodeInitialize}
Assume the conditions of Theorem~\ref{nodewiseSGD} and Lemma~\ref{nodeConv}. With high probability, we can initialize $(a_1^{(0)},\theta_1^{(0)})$ such that $\sqrt{L(a_1^{(0)},\theta_1^{(0)})} \leq \sqrt{L({0,0})} -\delta$ with $\delta = \frac{1}{\poly(d)}(d/\epsilon)^{ -18d}$ in time $\log(d)^{O(d)}$.
 \end{lemma}

\begin{proof}
By our conditions, there must exist some $|b_j|$ such that $|b_j| \geq 1/\poly(d)$ and for all $i$, $\|w_i - w_j\| \geq \Omega(d\log d)$. Note that if we randomly sample points in a ball of radius $O(d\log d)$, we will land in a $d$-neighborhood of $w_j$ with probability $\log(d)^{-O(d)}$ since $\|w_j\|\leq O(d\log d)$. 

Let $\theta_1$ be such that $\|\theta_1 - w_j \| \leq d$ and then we can solve for $a_1 = a_1^*(\theta_1)$ since we are simply minimizing a quadratic in one variable. Then, by Lemma~\ref{nodeGradient}, we see that $\|\nabla_{\theta_1}L_1 \| \geq 1/\poly(d)(d/\epsilon)^{-8d}$. Finally, by Lemma~\ref{almostHarmReal}, we know that the Hessian is bounded by $\poly(d)(d/\epsilon)^{2d }$. So, by Lemma~\ref{GradDecrease}, we conclude by taking a stepsize of $\alpha = \frac{1}{\poly(d)}(d/\epsilon)^{-2d}$ to reach $(a_1',\theta_1')$, we can guarantee that $L_1(a_1', \theta_1') \leq L_1 (a_1^*(\theta_1), \theta_1) - \frac{1}{\poly(d)} (d/\epsilon)^{ -18d}$.

But since $L_1(a_1^*(\theta_1),\theta_1)\leq L_1(0,0)$, we conclude that $L_1(a_1', \theta_1') \leq L_1(0,0) -  \frac{1}{\poly(d)} (d/\epsilon)^{ -18d}$. Let $\sqrt{L_1(a_1', \theta_1')} = \sqrt{L_1(0,0)} - \Delta \geq 0$. Squaring and rearranging gives $\Delta \geq \frac{1}{4\sqrt{L_1( 0, 0)}}  \frac{1}{\poly(d)}(d/\epsilon)^{ -18d} $. Since $L_1(0,0) \leq O(k) = O(\poly(d))$, we are done. 

\end{proof}
\begin{lemma}\label{nodewiseSGD}
Assume the conditions of Lemma~\ref{nodeConv}. Also, assume $b_1,...,b_k$ are any numbers in $[-1,1]$ and $w_1,...,w_k \in \R^d$ satisfy $\|w_i\|\leq O(d\log d)$ for all $i$ and there exists some $|b_j| \geq 1/\poly(d)$ with $\|w_i - w_j\| \geq \Omega(d\log d)$ for all $i$.

Then with high probability, we can choose an initial point $(a_1^{(0)}, \theta_1^{(0)})$ such that after running SecondGD (Algorithm \ref{SecondGD}) on the restricted regularized objective $L_1(a_1,\theta_1)$ for at most $(d/\epsilon)^{O(d)}$ iterations, there exists some $w_j$ such that $\|\theta_1 - w_j\| < \epsilon$. Furthermore, if $|b_j| \geq 1/\poly(d)$ and $\|w_i - w_j\| \geq \Omega(d\log d)$ for all $i$, then $\|\theta_1 - w_j\| < (d/\epsilon)^{-O(d)}$ and $|a + b_j| < (d/\epsilon)^{-O(d)}$.
\end{lemma}

\begin{proof}
First, by Lemma~\ref{nodeInitialize},  we can initialize ${(a_1^{(0)},\theta_1^{(0)})}$ such that $\sqrt{L_1({a_1^{(0)},\theta_1^{(0)}})} \leq \sqrt{ L_1({\bf 0,0})} - \delta$ for $ \delta = \frac{1}{\poly(d)}(d/\epsilon)^{ -18d}$. If we set $\alpha = (d/\epsilon)^{-O(d)}$ and $\eta = \gamma = \lambda \delta^2/(2d)$,  then running Algorithm~\ref{SecondGD} will terminate and return some $(a_1,\theta_1)$ in at most $(d/\epsilon)^{O(d)}$ iterations. This is because our algorithm ensures that our objective function decreases by at least $\min(\alpha \eta^2/2, \alpha^2\gamma^3/2)$ at each iteration and $G({\bf 0, 0})$ is bounded by $O(k)$ and $G \geq 0$ is non-negative.

Assume there does not exist $w_j$ such that $\|\theta_1 - w_j\| < (d/\epsilon)^{-O(d)}$. Then, we claim that $(a_1,\theta_1) \in \mathcal{M}_{L,\lambda \delta^2/(2d)}$. For the sake of contradiction, assume otherwise. By our algorithm termination conditions, then it must be that after one step of gradient or Hessian descent from $(a_1,\theta_1)$, we reach some $(a',\theta')$ and $L_1(a',\theta') > L_1(a_1,\theta_1) - \min(\alpha\eta^2/2,\alpha^2\gamma^3/2)$. Now, Lemma~\ref{almostHarmReal} ensures all first three derivatives of $\Phi$ are bounded by $(d/\epsilon)^{2d}$, except at $w_1,...,w_k$. Furthermore, since there does not exists $w_j$ such that $\|\theta_1 - w_j\| < (d/\epsilon)^{-O(d)}$, $L_1$ is three-times continuously differentiable within a $\alpha (d/\epsilon)^{2d} = (d/\epsilon)^{-O(d)}$ neighborhood of $ \theta_1$. Therefore, by Lemma~\ref{GradDecrease} and ~\ref{HessianDecrease}, we know that $L(a',\theta') \leq L_1(a',\theta') \leq L_1(a_1,\theta_1) - \min(\alpha\eta^2/2,\alpha^2\gamma^3/2)$, a contradiction. 

So, it must be $(a_1,\theta_1) \in \mathcal{M}_{L,\lambda \delta^2/(2d)}$. Since our algorithm maintains that our objective function is decreasing, so $\sqrt{L_1({a_1,\theta_1})} \leq \sqrt{L_1({\bf 0,0})} - \delta $. So, by Lemma~\ref{nodeRes}, there must be some $w_j$ such that $\|\theta- w_j\|\leq \epsilon$.

Now, if $|b_j| \geq 1/\poly(d)$ and $\|w_i - w_j\| \geq \Omega(d\log d)$ for all $i$, then since $(a,\theta) \in \mathcal{M}_{L,\lambda \delta^2/(2d)}$ and $\|\theta - w_j \| \leq \epsilon$, by Lemma~\ref{nodeGradient}, we have $\|\nabla_{\theta_1}L_1 \| \geq 1/\poly(d)(d/\epsilon)^{-8d} > \delta^2/(2d)$, a contradiction. Therefore, we must conclude that our original assumption was false and $\|\theta - w_j\| < (d/\epsilon)^{-O(d)}$ for some $w_j$.

Finally, we see that the charges also converge since $a = -2b_j \Phi_\epsilon(\|\theta - w_j\|) + O(d/\epsilon)^{-O(d)}$ and $\|\theta - w_j\| = (d/\epsilon)^{-O(d)}$. By noting that $\Phi_\epsilon(0) = 1$ and $\Phi_\epsilon$ is $O((d/\epsilon)^{2d})$-Lipschitz, we conclude. 
\end{proof}

Finally, we have our final theorem.

\nodewise*

\begin{proof}
Let our potential $\Phi_{\epsilon}$ be the one as constructed in Lemma~\ref{almostHarmReal} that is $\lambda$-harmonic for all $r \geq \epsilon$ with $\lambda = 1$. Let $(a_i, \theta_i)$ be the $i$-th node that is initialized and applied second order gradient descent onto. We want to show that the nodes $(a_i, \theta_i)$ will converge, in a node-wise fashion, to some permutation of $\{(b_1,w_1),...,(b_k,w_k)\}$. 

First, with high probability we know that $1 - 1/\poly(d) \geq |b_j| \geq 1/\poly(d)$ and $\|w_i\|\leq O(d\log d)$ and $\|w_i - w_j \| \geq \Omega(d\log d)$ for all $i, j$. By Lemma~\ref{nodewiseSGD}, we know that with high probability $(a_1,\theta_1)$ will converge to some $(d/\epsilon)^{-O(d)}$ neighborhood of $(b_{\pi(1)}, w_{\pi(1)})$ for some function $\pi: [k] \to [k]$. Now, we treat $a_1, \theta_1$ as one of the fixed charges and note that $|a_1| \leq 1$ and $\|\theta_1\| \leq O(d\log d)$ and as long as $k > 1$ (if $k = 1$, we are done), then there exists $|b_j| \geq 1/\poly(d)$ with $\|w_i - w_j\| \geq \Omega(d\log d)$ for all $i$ and $\|\theta_1 - w_j \| \geq \Omega(d\log d)$.

Then, by Lemma~\ref{nodeInitialize}, we can initialize $(a_2^{(0)}, \theta_2^{(0)})$ such that $\sqrt{L_2(a_2^{(0)}, \theta_2^{(0)})} \leq \sqrt{L_2(0,0)} - \delta$, with $\delta = 1/\poly(d) (d/\epsilon)^{-18d}$. Then, by Lemma~\ref{nodewiseSGD}, we know that $(a_2, \theta_2)$ will converge to some $w_{\pi(2)}$ such that $\|\theta_2 - w_{\pi(2)}\| <\epsilon$ (or $\|\theta_2 - \theta_1 \| < \epsilon$ but $\theta_2$ is still $\epsilon$-close to $w_{\pi(1)}$). We claim that $\pi(1) \neq \pi(2)$.

By optimality conditions on $a_2$, we see that 
\begin{align*}
a_2^*(\theta_2) = a_1 \Phi_{\epsilon}(\|\theta_2 - \theta_1\|) + b_j \Phi_{\epsilon}(\|\theta_1 - w_j\|) +  \sum_{i \neq j} b_i \Phi_{\epsilon}(\|\theta_1 - w_i\|)
\end{align*}
If $w_{\pi(1)} = w_{\pi(2)}$, then note that $\|\theta_1 - w_i \| \geq \Omega(d\log d)$ for all $i \neq \pi(1)$. Therefore, $2 \sum_{i \neq j} b_i \Phi_{\epsilon}(\|\theta_1 - w_i\|) = (d/\epsilon)^{-O(d)}$. And by our convergence guarantees and the $(d/\epsilon)^{2d}$-Lipschitzness of $\Phi_{\epsilon}$, $ a_1 \Phi_{\epsilon}(\|\theta_2 - \theta_1\|) + b_j \Phi_{\epsilon}(\|\theta_1 - w_j\|) \leq (d/\epsilon)^{-O(d)}$. Therefore, $a_2^*(\theta_2) \leq (d/\epsilon)^{-O(d)}$. 

However, we see that $L_2(a_2, \theta_2) \geq L_2(a_2^*(\theta_2),\theta_2) = L_2(0,0) - \frac{1}{2}a_2^*(\theta_2)^2 \geq L_2(0,0) - (d/\epsilon)^{-O(d)}$. But since $L_2$ is non-increasing, this contradicts our initialization and therefore $\pi(1) \neq \pi(2)$. Therefore, our claim is done and by Lemma~\ref{nodewiseSGD}, we see that since $|b_{\pi(2)}| \geq 1/\poly(d)$ and for all $i$, $\|w_i - w_{\pi(2)}\| \geq \Omega(d \log d)$ and $\|\theta_1 - w_{\pi(2)}\| \geq \Omega(d\log d)$, we conclude that $(a_2,\theta_2)$ is in a $(d/\epsilon)^{-O(d)}$ neighborhood of $(b_{\pi(2)}, \theta_{\pi(2)})$. Finally, we induct and by similar reasoning, $\pi$ is a permutation. Now, our theorem follows. 
\end{proof}

\section{Convergence of Almost Strictly Subharmonic Potentials}\label{App:Subharm}

\begin{definition}
$\Phi(\theta,w)$ is a {\bf strictly subharmonic} potential on $\Omega$ if it is differentiable and $\Delta_\theta \Phi(\theta,w) > 0$ for all $\theta \in \Omega$, except possibly at $\theta = w$.
\end{definition}

An example of such a potential is $\Phi(\theta, w) = \|\theta -w \|^{2-d-\epsilon}$ for any $\epsilon > 0$. Although this potential is unbounded at $\theta = w$ for most $d$, we remark that it is bounded when $d = 1$. Furthermore, the sign of the output weights $a_i, b_i$ matter in determining the sign of the Laplacian of our loss function. Therefore, we need to make suitable assumptions in this framework.

Under Assumption \ref{outputFixed}, we are working with an even
simpler loss function:
\begin{equation}\label{errFixed}
L(\theta) =  2\sum_{i=1}^k\sum_{i < j} \Phi(\theta_i,\theta_j) - 2\sum_{i=1}^k\sum_{j=1}^k\Phi(\theta_i,w_j)
\end{equation}
\Anote{i don't think you should show this theorem. you ahve more significant results, and this distracts from them. plus it doesn't refer to any realistic learning setup. it's an intermediate warmup illustration for yourself, not for the reader.}
\begin{theorem}\label{subStrict}
Let $\Phi$ be a symmetric strictly subharmonic potential on $\mathcal{M}$ with $\Phi(\theta,\theta) = \infty$. Let Assumption \ref{outputFixed} hold and let $L$ be as in \eqref{errFixed}. Then, $L$ admits no local minima, except when $\theta_i = w_j$ for some $i, j$.
\end{theorem}
\begin{proof}
First, let $\Phi$ be translationally invariant and $\mathcal{M} =
\R^d$. Let $\boldsymbol{\theta}$ be a critical point. Assume, for sake
of contradiction, that for all $i, j$, $\theta_i \neq w_j$. If
$\theta_i$ are not distinct, separating them shows that
we are not at a local minima since $\Phi(\theta_i,\theta_j) = \infty$
and finite elsewhere. 

The main technical detail is to remove interaction terms between
pairwise $\theta_i$ by considering a correlated movement, where each
$\theta_i$ are moved along the same direction $v$. In this case,
notice that our objective, as a function of $v$, is simply
\begin{align*}
& H(v) = L(\theta_1+ v, \theta_2 + v, ...,\theta_k + v) \\
& =  2\sum_{i=1}^k\sum_{i < j} \Phi(\theta_i+v,\theta_j+v) -
  2\sum_{i=1}^k\sum_{j=1}^k \Phi(\theta_i+v,w_j)
\end{align*}

Note that the first term is constant as a function of $v$, by
translational invariance. Therefore,
\[\nabla_{v}^2 H = -2\sum_{i=1}^k \sum_{j=1}^k \nabla^2\Phi(\theta_i, w_j)\]
By the subharmonic condition,
$\Tr(\nabla_{v}^2H) = -2\sum_{i=1}^k\sum_{j=1}^k
\Delta_{\theta_i}\Phi(\theta_i,w_j) < 0$.
Therefore, we conclude that $\theta$ is not a local minima of $H$ and
$L$.  We conclude that $\theta_i = w_j$ for some $i, j$.

The above technique generalizes to $\Phi$ being rotationally invariant
case by working in spherical coordinates and correlated translations
are simply rotations. Note that we can change to spherical coordinates
(without the radius parameter) and let
$\widetilde{\theta_1},...,\widetilde{\theta_{k}}$ be the standard
spherical representation of $\theta_1,...,\theta_k$.

We will consider a correlated translation in the spherical coordinate space, which are simply rotations on the sphere. Let $v$ be a vector in $\R^{d-1}$ and our objective is simply
\[ H(v) = L( \widetilde{\theta_1}+v,...,\widetilde{\theta_{k}} +v)\]

Then, we apply the same proof since $\Phi( \widetilde{\theta_i}+v, \widetilde{\theta_j}+v)$ is constant as a function of $v$ by rotationally invariance.
\end{proof}
\begin{corollary}
Assume the conditions of Theorem \ref{subStrict} and $\Phi(\theta,\theta) < \infty$. Then, $L$ admits no local minima, except at the global minima.
\end{corollary} 

\begin{proof}
From the same proof from theorem \ref{subStrict}, we conclude that there must exists $i, j$ such that $\theta_i = w_j$. Then, since $\Phi(\theta,\theta) < \infty$, notice that $\theta_i, w_j$ cancels each other out and by drop $\theta_i, w_j$ from the loss function, we have a new loss function $L$ with $k-1$ variables. Then, using induction, we see that $\theta_i = w_{\pi(i)}$ at the local minima for some permutation $\pi$.
\end{proof}

For concreteness, we will focus on a specific potential function with this property: the Gaussian kernel $\Phi(\theta, w) = \exp(-\|\theta - w\|^2/2)$. In $\R^d$, the Laplacian is $\Delta \Phi = ( \|\theta - w\|^2 -d ) \exp(-\|\theta - w\|^2/2)$, which becomes positive when $\|\theta - w \|^2 \geq d$. Thus, $\Phi$ is strictly subharmonic outside a ball of radius $\sqrt{d}$. This informally implies that $\theta_1$ converges to a $\sqrt{d}$-ball around some $w_j$.

For concreteness, we will focus on a specific potential function with
this property: the Gaussian kernel $\Phi(\theta, w) = \exp(-c\|\theta
- w\|^2/2)$, which corresponds to a Gaussian activation. In $\R^d$, the Laplacian is $\Delta \Phi = ( c\|\theta - w\|^2 -d ) \exp(-c\|\theta - w\|^2/2)$, which becomes positive when
$\|\theta - w \|^2 \geq d/c$. Thus, $\Phi$ is strictly subharmonic
outside a ball of radius $\sqrt{d/c}$. Note that Gaussian potential
restricted to $S^{d-1}$ gives rise to the exponential activation
function, so we can show convergence similarly.  
\begin{theorem}\label{gaussStrict}
\label{GaussStrict}
Let $\mathcal{M} = \R^d$ and $\Phi(\theta,w) = e^{-c\|\theta-w\|^2/2}$ and Assumption \ref{outputFixed} holds. Let $L$ be as in \eqref{errFixed} and $\|\boldsymbol{w}\|\leq poly(d)$. 

If $c = O(d/\epsilon)$ and $(\boldsymbol{a,\theta}) \in \mathcal{M}_{e^{-\poly(d,1/\epsilon)}}$, then there exists $i, j$ such that $\| \theta_i - w_j \|^2 \leq \epsilon$.
\end{theorem}

\begin{proof}
Consider again a correlated movement, where each $\theta_i$ are moved along the same direction $v$. As before, this drops the pairwise $\theta_i$ terms. If for all $i, j$ $\| \theta_i - w_j \|^2 \leq \epsilon$, then we see that $\Delta_{\theta_i} \Phi = ( c\|\theta - w\|^2 -d ) \exp(-c\|\theta - w\|^2/2) > e^{-poly(d,1/\epsilon)}$. 
\[\Tr(\nabla^2 L) = -2\sum_{i=1}^k \sum_{j=1}^k \Delta_{\theta_i}\Phi(\theta_i, w_j) < -e^{-poly(d,1/\epsilon)}\]

Therefore, $\nabla^2 L$ must admit a strictly negative eigenvalue that
is less than $e^{-c_3 d}$, which implies our claim (we drop the
$\poly(d,k)$ terms).

\end{proof}

\section{Common Activations}
First, we consider the sign activation function. Under restrictions on the size of the input dimension or the number of hidden units, we can prove convergence results under the sign activation function, as it gives rise to a harmonic potential.

\begin{assumption}
\label{outputFixed}
All output weights $b_i = 1$ and therefore the output weights  $a_i = - b_i = -1$ are fixed throughout the learning algorithm. 
\end{assumption}

\begin{restatable}{lemma}{signcon}\label{signCon}
Let $\mathcal{M} = S^1$ and let Assumption \ref{outputFixed} hold. Let $L$ be as in \eqref{errSimp} and $\sigma$ is the sign activation function. Then $L$ admits no strict local minima, except at the global minima.
\end{restatable}

We cannot simply analyze the convergence of GD
on all $\theta_i$ simultaneously since as before, the pairwise
interaction terms between the $\theta_i$ present complications. Therefore, we now only consider the convergence guarantee of gradient descent on the first node, $\theta_1$, to some $w_j$, while the other nodes are inactive (i.e. $a_2,...,a_k = 0$). In essence, we are working with the following simplified loss function.
\begin{equation}\label{errLossUnit}
L(a_1,\theta_1) =  a_1^2 \Phi(\theta_1,\theta_1)  + 2\sum_{j=1}^k a_1b_j \Phi(\theta_1,w_j)
\end{equation}

\begin{restatable}{lemma}{signconv}
\label{SignConv}
Let $\mathcal{M} = S^1$ and $L$ be as in \eqref{errLossUnit} and $\sigma$ is the sign activation function. Then, almost surely over random choices of $b_1,...,b_k$, all local minima of $L$ are at $\pm w_j$. 
\end{restatable}
For the polynomial activation and potential functions, we also can show convergence under orthogonality assumptions on $w_j$. Note that the realizability of polynomial potentials is guaranteed in Section~\ref{sec:realizable}.

\begin{restatable}{theorem}{polystrict}
\label{PolyStrict}
Let $\mathcal{M} = S^{d-1}$. Let $w_1,...,w_k$ be orthonormal vectors in $\R^d$ and $\Phi$ is of the form $\Phi(\theta,w) = (\theta^Tw)^l$ for some fixed integer $l \geq 3$. Let $L$ be as in \eqref{errLossUnit}. Then, all critical points of $L$ are not local minima, except when $\theta_1 = w_j$ for some $j$.   
\end{restatable}

\subsection{Convergence of Sign Activation}

\signcon*

\begin{proof}
We will first argue that unless all the electrons and protons have matched up as a permutation it cannot be a strict local minimum and then argue that the global minimum is a strict local minimum.

First note that if some electron and proton have merged, we can remove such pairs and argue about the remaining configuration of charges. So WLOG we assume there are no such overlapping electron and proton.

First consider the case when there is an isolated electron $e$ and there is no charge diagonally opposite to it. In this case look at the two semicircles on the left and the right half of the circle around the isolated electron -- let $q_1$ and $q_2$ be the net charges in the left and the right semi-circles. Note that $q_1 \neq q_2$ since they are integers and $q_1 + q_2 = +1$ which is odd. So by moving the electron slightly to the side with the larger charge you decrease the potential.

If there is a proton opposite the isolated electron the argument becomes simpler as the proton benefits the motion of the electron in either the left or right direction. So  the only way the electron does not benefit by moving in either direction is that $q_1 = -1$ and $q_2 = -1$ which is impossible.

If there is an electron opposite the isolated electron then the combination of these two diagonally opposing electrons have a zero effect on every other charge. So it is possible rotate this pair jointly keeping them opposed in any way and not change the potential. So this is not a strict local minimum.

Next if there is a clump of isolated electrons with no charge on the diagonally opposite point then again as before if $q_1 \neq q_2$ we are done. If  $q_1 = q_2$ then the the electrons in the clump locally are unaffected by the remaining charges. So now by splitting the clump into two groups and moving them apart infinitesimally we will decrease the potential.

Now if there is only protons in the diagonally opposite position an isolated electron again we are done as in the case when there is one electron diagonally opposite one proton. 

Finally if there is only electrons diagonally opposite a clump of electrons again we are done as we have found at least one pair of opposing electrons that can be jointly rotated in any way.

Next we will argue that a permutation matching up is a strict local minumum. For this we will assume that no two protons are diagonally opposite each other (as they can be removed without affecting the function). Now given a perfect matching up of electrons and protons, if we perturb the electrons in any way infinitesimally, then any isolated clump of electrons can be moved slightly to  the left or right to improve the potential.
\end{proof}

\signconv*

\begin{proof}
In $S^1$, notice that the pairwise potential function is $\Phi(\theta,w) = 1 - 2\cos^{-1}(\theta^Tw)/\pi = 1 - 2\alpha/\pi$, where $\alpha$ is the angle between $\theta, w$. So, let us parameterize in polar coordinates, calling our true parameters as $\widetilde{w_1},...,\widetilde{w_k} \in [0,2\pi]$ and rewriting our loss as a function of $\widetilde{\theta} \in [0,2\pi]$. 

Since $\Phi$ is a linear function of the angle between $\theta, w_j$, each $w_j$ exerts a constant gradient on $\widetilde{\theta}$ towards $\widetilde{w_j}$, with discontinuities at $\widetilde{w_j},\pi+\widetilde{w_j}$. Almost surely over $b_1,..,b_k$, the gradient is non-zero almost everywhere, except at the discontinuities, which are at $\widetilde{w_j}, \pi+\widetilde{w_j}$ for some $j$. 
\end{proof}

\subsection{Convergence of Polynomial Potentials}

\polystrict*

\begin{proof}
WLOG, we can consider $w_1,...,w_d$ to be the basis vectors $e_1,...,e_d$. Note that this is a manifold optimization problem, so our optimality conditions are given by introducing a Lagrange multiplier $\lambda$, as in \cite{GeHJY15}.
\[\pd{L}{a} = 2\sum_{i=1}^d ab_i (\theta_i)^l + 2a = 0\]
\[ (\nabla_\theta L)_i = 2ab_il(\theta_i)^{l-1}  -2\lambda \theta_i = 0 \]
where $\lambda$ is chosen that minimizes 
\[\lambda = \arg \min_\lambda \sum_i (ab_i l (\theta_i)^{l-1} - \lambda\theta_i)^2 = \sum ab_i l (\theta_i)^l \]
Therefore, either $\theta_i = 0$ or $b_i (\theta_i)^{l-2} = \lambda/(al)$. From \cite{GeHJY15}, we consider the constrained Hessian, which is a diagonal matrix with diagonal entry: 
\[(\nabla^2 L)_{ii} = 2a b_i l(l-1)(\theta_i)^{l-2} - 2 \lambda\]
Assume that there exists $\theta_i, \theta_j \neq 0$, then we claim that $\theta$ is not a local minima. First, our optimality conditions imply $b_i(\theta_i)^{l-2} = b_j (\theta_j)^{l-2} = \lambda/(al)$. So,
\[(\nabla^2 L)_{ii} = (\nabla^2L)_{jj} = 2a b_i l(l-1)(\theta_i)^{l-2} - 2 \lambda\]
\[ = 2(l-2)\lambda = -2(l-2)la^2\]
Now, there must exist a vector $v \in S^{d-1}$ such that $v_k = 0$ for $k \neq i,j$ and $v^T\theta = 0$, so $v$ is in the tangent space at $\theta$. Finally, $v^T(\nabla^2 L) v  = -2(l-2)l a^2 < 0$, implying $\theta$ is not a local minima when $a \neq 0$. Note that $a = 0$ occurs with probability 0 since our objective function is non-increasing throughout the gradient descent algorithm and is almost surely initialized to be negative with $a$ optimized upon initialization, as by observed before.
\end{proof}

Under a node-wise descent algorithm, we can show polynomial-time convergence to global minima under orthogonality assumptions on $w_j$ for these polynomial activations/potentials. We will not include the proof but it follows from similar techniques presented for nodewise convergence in Section~\ref{App:EigenFunc}.

\section{Proof of Sign Uniqueness}
For the sign activation function, we can show a related result.
\begin{restatable}{theorem}{signUnique}
\label{SignUnique}
Let $\mathcal{M} = S^{d-1}$ and $\sigma$ be the sign activation function and $b_2,...,b_k = 0$. If the loss \eqref{errLoss} at $(\boldsymbol{a,\theta})$ is less than $O(1)$, then there must exist $\theta_i$ such that $w_1^T\theta_i > \Omega(1/\sqrt{k})$.
\end{restatable}
\begin{proof}
WLOG let $w_1 = e_1$. Notice that our loss can be bounded below by Jensen's:
\begin{align*}
& \expt_X \left[ \left( \sum_{i=1}^k a_i \sigma(\theta_i^TX) - \sigma(X_1)\right)^2 \right] \\
& \qquad 
\geq \expt_{X_1} \left[ \left( \EE{X_2...X_d}{\left[ \sum_{i=1}^k a_i \sigma(\theta_i^TX) \right]}- \sigma(X_1)\right)^2 \right],
\end{align*}
where $X$ is a standard Gaussian in $\R^d$. 
\begin{align*}
E_{X_2,..,X_d} \left[  \sum_{i=1}^k a_i \sigma(\theta_i^TX) \right] &= \sum_{i=1}^k a_i E_{X_2,...X_d}\left[  \sigma(\theta_{i1}X_1 + \sum_{j >1} \theta_{ij}X_{j})  \right]\\
&= \sum_{i=1}^k E_{Y} \left[   \sigma(\theta_{i1}X_1 + \sqrt{1-\theta_{i1}^2}Y)  \right]  \\
&= \sum_{i=1}^k a_i E_{Y} \left[   \sigma(\textstyle\frac{\theta_{i1}}{\sqrt{1-\theta_{i1}^2}}X_1 + Y)  \right] ,
\end{align*}
where $Y$ is an independent standard Gaussian and for any small $\delta$, if $p(y)$ is the standard Gaussian density, 
\[ E_Y[\sigma(\delta + Y)] = \int_{-\delta}^{\delta} p(y) \, dy = 2p(0)\delta + O(\delta^2) \]

If $w_1^T\theta_i = \theta_{i1} < \epsilon$ for all $i$, then notice that with high probability on $X_1$ (say condition on $|X_1| \leq 1$), 
\[\expt_{Y} \left[   \sigma(\textstyle\frac{\theta_{i1}}{\sqrt{1-\theta_{i1}^2}}X_1 + Y)  \right] = 2p(0)\textstyle\frac{\theta_{i1}}{\sqrt{1-\theta_{i1}^2}}X_1 + O(\epsilon^2X_1^2)\]

Therefore, since $\epsilon < O(1/\sqrt{k})$,
\begin{align*}
\expt_{X_2,..,X_d} \left[  \sum_{i=1}^k a_i \sigma(\theta_i^TX)
  \right]  & = X_1
  \sum_{i=1}^k2p(0)a_i\textstyle\frac{\theta_{i1}}{\sqrt{1-\theta_{i1}^2}}
  + O(k\epsilon^2X_1^2) \\
& = cX_1+O(1)
\end{align*}

Finally, our error bound is now
\begin{align*}
& \expt_{X_1} \left[ \left( \expt_{X_2...X_d}\left[ \sum_{i=1}^k a_i
      \sigma(\theta_i^TX) \right]- \sigma(X_1)\right)^2 \right] \\
& \qquad \geq
\expt_{|X_1| \leq 1}[(cX_1+O(1) - \sigma(X_1))^2]
\end{align*}

And the final expression is always larger than some constant, regardless of $c$.
\end{proof}


\if{1}

\subsection{Infinite Iteration Bounds} 
\label{InfIter}

\begin{theorem}\cite{lee2016gradient, PanageasP16}\label{convStrict}
  Let $f :\Omega \to \R$ be a twice differentiable function such that
  $\sup_{x \in \Omega} \|\nabla^2 f\| \leq L$. Let
  $\mathcal{S} \subseteq \Omega$ be the set of critical points of $f$
  that are not local minima. Also, if
  $g(x) = x - \frac{1}{2L} \nabla f(x)$, then
  $g(\Omega) \subseteq \Omega$.

  Then, running Algorithm \ref{GD} with gradient input $\nabla f$ and
  stepsize $\alpha = 1/(2L)$, as the iteration $T \to\infty$, will
  converge to a point $x_\infty$ outside of $S$ almost surely over
  randomly chosen initial points $x_0$.
\end{theorem}

\begin{corollary}
Assume all the assumptions of Theorem \ref{convStrict} and let $f$ admit a global minima in $\overline{\Omega}$. Assume all critical points of $f$ in $\Omega$ are not local minima, except at the global minima. Then, running Algorithm \ref{GD} with gradient input $\nabla f$ and stepsize $\alpha = 1/(2L)$ will converge to the global minima almost surely as the iteration count $T \to\infty$.
\end{corollary}
\fi

\if{1}
\begin{theorem}
For any $\epsilon < 1/\poly(d)$, we can construct a realizable potential $\Phi$ such that with high probability, running Algorithm \ref{NodeGDOpt} on \eqref{errLoss} with error $\delta = \poly(\epsilon,1/d)$, $\gamma = \epsilon$ and stepsize $\alpha = 1/\poly(d,1/\epsilon)$ converges in $T = \poly(d, 1/\epsilon)$ iterations to $(\boldsymbol{a,\theta})$ such that either  $\theta$ is within $\epsilon$-neighborhood of the global minima or there exists $i$ such that if $\theta_i$ is picked uniformly in $\mathcal{M}$
\[ \expt\left[\left( \sum_{j < i} a_j \Phi(\theta_i,\theta_j) + \sum_{j=1}^k b_j \Phi(\theta_i,w_j)\right)^2\right] < \epsilon\]

The sample complexity is $d^{O(\log(d)/\epsilon)}$.
\end{theorem}

\begin{proof}
Let $\Phi_m$ be the $(1,m)$-Harmonic potential in Theorem \ref{eigConv} with $m = \poly(d,1/\epsilon)$ and $m$ odd. We first consider the algorithm on node $\theta_1$ and claim that it will merge with some $w_j$ and then we will proceed with induction.

If
\[ \expt\left[\left( \sum_{j < 1} a_j \Phi(\theta_i,\theta_j) +
    \sum_{j=1}^k b_j \Phi(\theta_i,w_j)\right)^2\right] < \epsilon,\]
then we are done. Otherwise, with high probability, Theorem
\ref{nonDecrease} allows us to deduce that throughout the SGD
algorithm applied on $\theta_1$, $a_1^2 = \Omega(\epsilon)$.

Now we want to apply Theorem \ref{strongConverge}, so we check the regularity conditions. Since $\mathcal{M} = S^{d-1}$, then we can choose $B, L, \rho$ to be $\poly(d)$ since $\Phi$ and the second and third partials of $\Phi$ are all bounded by $\poly(d)$. Furthermore, by our construction, our activation function $\sigma(x)$ and its derivatives are $O(|x|^{\poly(d,1/\epsilon)})$. By Theorem \ref{genErrBound}, with high probability, we can construct a stochastic oracle up to $\poly(\epsilon,1/d)$ error with sample complexity $d^{\poly(d,1/\epsilon)}$.

Therefore, by Theorem \ref{strongConverge} we conclude that we converge to $\theta_1 \in \mathcal{M}_{\poly(\epsilon,1/d)}$. By Theorem \ref{eigConv}, since $|a_i| = \Omega(\sqrt{\epsilon})$, this implies that it is in an $\poly(\epsilon,1/d)$-neighborhood of some $w_{i}$ in $\poly(d,1/\epsilon)$ time. Note that $\theta_1$ will close to with $\pm w_j$ for some $j$ but since $\Phi_m$ is odd, WLOG, it is close to $w_j$. 

Furthermore, note that $|a_1| \leq \poly(d)$ by using the explicit formula. And lastly, by Theorem \ref{quadConverge}, since the maximum eigenvalue of our matrix $A$ is bounded by \poly(d), our gradient descent steps on the quadratic loss $L_{a_1}$ will converge to the optimum in $\Omega = \{a \in \R^n | \|a\| \leq \poly(d)\}$ with $O(\epsilon)$ error in $T$ iterations.

Now, we proceed with induction and repeat the same argument on $\theta_2$. We can simply treat $\theta_1$ as $w_{k+1}$ and so applying the same argument tells us that $\theta_2$ is close to some $w_j$ for some $j$. The issue is that $\theta_2$ could be in a $\poly(\epsilon,1/d)$-neighborhood of $w_{k+1} = \theta_1$ or $w_i$. We claim that this will not occur. First, since $w_i, w_{k+1}$ are in a $\poly(\epsilon,1/d)$-neighborhood of each other, we will assume WLOG that $\theta_2$ is close to $\theta_1$.

Now, by the optimality of $a_1$, we know that $L(a_1,\theta_1) \leq \min_{a \in \Omega} L(a,\theta_1) + O(\epsilon)$. We claim that if $\theta_2$ is close to $\theta_1$, then $L(a_1+a_2,\theta_1) \leq L(a_1,\theta_1) - \Omega(\epsilon)$. This, combined with the fact that $a_1 + a_2$ is bounded by \poly(d), would lead to a contradiction.

First, notice that since $a_2$ is always optimal, we have
\begin{align*}
& L(a_1,\theta_1) - L(a_1,\theta_1,a_2,\theta_2) \\
& \qquad \qquad = a_2^2 + 2a_2 (\sum_{j < 2} \Phi(\theta_2,\theta_j) + \sum_{j=1}^k b_j \Phi(\theta_2,w_j)) \\
& \qquad \qquad = -a_2^2 = \Omega(\epsilon)
\end{align*}

Therefore, it suffices to show that $|L(a_1+a_2,\theta_1) - L(a_1,\theta_1,a_2,\theta_2) | \leq O(\epsilon)$. Since $L(a_1+a_2,\theta) = L(a_1,\theta_1,a_2,\theta_1)$, this follows immediately from the $\poly(d)$-Lipschitz of $\Phi$ and the fact that $\theta_2$ is in a $\poly(d,1/\epsilon)$-neighborhood of $\theta_1$. We conclude that $\theta_2$ cannot be in a $\poly(\epsilon,1/d)$-neighborhood of $\theta_1$ but converges to a point close to $w_{j}$, $j\neq i,k+1$. Therefore, the no two $\theta_i$ are matched to one $w_j$. 

By applying this logic to all $\theta_i$ through induction, we deduce that $\theta$ is within $\epsilon$ of the global minima.
\end{proof}
\fi

\end{document}